\documentclass[%
 aip,
 amsmath,amssymb,floatfix,
 reprint,%
]{revtex4-1}

\usepackage{graphicx}
\usepackage{dcolumn}
\usepackage{bm}

\usepackage[utf8]{inputenc}
\usepackage[T1]{fontenc}
\usepackage{mathptmx}
\usepackage{etoolbox}

\usepackage{dsfont}
\usepackage{mathrsfs}
\usepackage{tikz}
\usepackage{graphicx} 
\usepackage{float}
\usepackage{wrapfig}
\usepackage{amsthm}

\newtheorem{mydef}{Definition} 
\newtheorem{mythm}{Theorem}
\newtheorem{mylem}{Lemma}
\newtheorem{myprop}{Proposition}

\newcommand{\sotimes}{\widehat{\otimes}} 
\newcommand{\grtrans}{\mathbb{E}} 
\newcommand{\grtranspose}{\textup{t}} 
\newcommand{\virtHilb}{H} 
\newcommand{\Tr}{\textup{Tr}} 

\makeatletter
\def\@email#1#2{%
 \endgroup
 \patchcmd{\titleblock@produce}
  {\frontmatter@RRAPformat}
  {\frontmatter@RRAPformat{\produce@RRAP{*#1\href{mailto:#2}{#2}}}\frontmatter@RRAPformat}
  {}{}
}%
\makeatother
\begin{document}
	
	\bibliographystyle{siam}

\preprint{AIP/123-QED}

\title[]{Non-local order parameters\\for fermion chains via the partial transpose}
\author{L. P. Mayer}
 \affiliation{Physics Department, University of Cologne.}
 \email{lmayer@thp.uni-koeln.de.}

\date{\today}

\begin{abstract}
In the last two decades, a vast variety of topological phases have been described, predicted, classified, proposed, and measured. While there is a certain unity in method and philosophy, the phenomenology differs wildly. This work deals with the simplest such case: fermions in one spatial dimension, in the presence of a symmetry group $G$ which contains anti-unitary symmetries. A complete classification of topological phases, in this case, is available. Nevertheless, these methods are to some extent lacking as they generally do not allow to determine the class of a given system easily. This paper will take up proposals for non-local order parameters defined through anti-unitary symmetries. They are shown to be homotopy invariants on a suitable set of ground states. For matrix product states, an interpretation of these invariants is provided: in particular, for a particle-hole symmetry, the invariant determines a real division super algebra $\mathbb{D}$ such that the bond algebra is a matrix algebra over $\mathbb{D}$.
\end{abstract}

\maketitle

\section{Introduction}
Thermodynamical phases are stable in that their properties generically do not depend strongly on variations of external parameters -- a glass of water has roughly the same characteristics at $14^\circ$C as at $15^\circ$C. This can be formalized as the statement that correlation functions generically depend continuously on external parameters $\lambda$.
Recently, interest mounted in a stronger notion of stability. The systems are quantum, gapped, and at $0$K, i.e., described -- when in equilibrium -- by the ground state $\Omega$ of the Hamiltonian $H$. Denote the set of all ground states as $\mathcal{T}$. Consider slow variations in parameter space, $[0,1] \ni t \mapsto \lambda(t)$. If $H(\lambda(t))$ is gapped and $\Omega(\lambda(t))$ unique for all $t$, the time evolution generated by $H(\lambda(t))$ will carry $\Omega(\lambda(0))$ to $\Omega(\lambda(1))$. This defines an equivalence relation denoted as $\Omega(\lambda(0)) \sim \Omega(\lambda(1))$. The equivalence classes are called \emph{topological phases}\footnote{This work is in one spatial dimension, where \emph{long-range-entangled}(LRE) or \emph{symmetry protected} (SPT) and \emph{short-range entangled}(SRE) or \emph{topologically ordered} topological phases can be treated simultaneously\cite{Chen2010,Fidkowski2011}. These distinctions become more relevant in higher dimensions.} -- they are stable under continuous deformations. In many contexts, these classes are too large. Physical system exist in a definite spatial dimension $d$. They might have a symmetry group $G$. As a sort of regularization procedure, impose a uniform bound $\ell_c$ on the correlation length\footnote{Another approach is to impose regularity conditions on the path directly\cite{Bachmann2015}. Another interesting restriction, which is not pursued here, is that the microscopic degrees of freedom are kept constant along the path\cite{kennedy2016topological}.}. The set of states abiding by these constraints is denoted as $\mathcal{T}_{\ell_c}^{G|d}$. Even with this more limited set of deformations, it is maybe surprising that there are distinct topological phases. To see what this entails, consider a path -- in the set of all $G$-symmetric states in $d$ dimensions, a larger set than $\mathcal{T}_{\ell_c}^{G|d}$ -- connecting two states crossing a quantum phase transition. If the two states are in different topological phases, \emph{any} path connecting them will contain at least one quantum phase transition point. The reverse is false; even if two states can be connected through a quantum phase transition, it might also be possible to find a non-critical connecting path.
\\
In one\cite{Bourne_2021} and two\cite{naaijkens2021split,ogata2021classification} dimensions, such phases have been classified. Building on earlier work within the matrix product states (MPS) paradigm\cite{Kapustin2017,Kapustin2018},a cohomology class $\nu$ is associated to a quantum state by virtually cutting the system in half. $\nu$ is determined by algebraic properties of the half-space operator algebra as represented on the Hilbert space given through the Gel'fand-Naimark-Segal construction. This \emph{does} give an explicit way of determining the class of any given state, which, moreover, is fairly easy to do when the state is given as a matrix product state respecting the symmetry. However, if the state is presented differently, it first has to be brought into tensor network form\cite{pollmann2017symmetry}. As the invariants are algebraic, further difficulties might arise when the symmetry is only approximate, or even -- as in numerical schemes -- the presentation preserves symmetries of the Hamiltonian only approximately. Such problems can be remedied by using \emph{order parameters} to characterize the phase a state is in. For bosonic topological phases, string order parameters are known to extract cohomological invariants of MPS\cite{Pollmann2012Detec}. There are also suggestions for generalized string order parameters useful for anti-unitary symmetries in fermion systems\cite{Ryu2017}. In this case, the cohomological invariants play a slightly different r\^{o}le. Bond algebras of pure state fermionic MPS (fMPS) are graded simple, i.e., isomorphic to complex Clifford algebras. \emph{Unitary} symmetries do not affect this, but \emph{anti-unitaries} can. In particular, if the anti-unitary transformation squares to the identity on physical Hilbert space, the bond algebras are promoted to \emph{real} Clifford algebras\cite{Fidkowski2011}. The structure can change by adding unitary symmetries\cite{geiko2021dyson}. Reflections, which are unitary operations that act not on-site -- beyond the scope of this work -- behave in this similarly to anti-unitary operations\cite{Bultinck_2017}, as can be rationalized through an appeal to the continuum\cite{Kapustin2015}.
\\
This work adds to these results in two directions: (a) the connection between the cohomological invariants for fMPS and the generalized string order parameters is established; and (b) this is used to define generalized string order parameters \emph{beyond} MPS. The second part relies strongly on MPS as an intermediate link in the argument. Hence (c) an exposition on fMPS from the perspective of states as functionals on an operator algebra. Compared to vector approaches\cite{Bultinck_2017} this allows working directly in the thermodynamic limit.

\section{Fermion Chains}

The starting point is the description of one-dimensional fermions in terms of creation and annihilation operators. To more smoothly include particle-hole symmetries, this will be done here starting from a Nambu space $\mathscr{W}$ endowed with a bilinear form $\{ \cdot,\cdot\}$ and a bracket-preserving real structure $\rho$. The relevant case is that of a chain of atoms with $p$ orbitals, the single particle motion is described by a space $\mathscr{V} := \ell^2(\mathbb{Z};\mathbb{C}^p)$. The space of orbitals will be denoted by $V$. Then $\mathscr{W} = \mathscr{V} \oplus \mathscr{V}^* = \ell^2(\mathbb{Z},W)$ with $W = V \oplus V^*$ and 
\begin{align*}
\{ f_1 + \varphi_1 , f_2 + \varphi_2 \} &:= \varphi_1(f_2) + \varphi_2(f_2) \ , \\
 \quad \rho(f + \langle g, \cdot \rangle) &:= g + \langle f,\cdot \rangle \ .
\end{align*}
Second quantization is achieved by introducing abstract operators $\gamma(w)$ for $w \in \mathscr{W}$ satisfying
\begin{align}
\begin{aligned}
\gamma(\alpha w + \beta w') &= \alpha \gamma(w) + \beta \gamma(w') \ , \\ 
\gamma(w) \gamma(w') + \gamma(w') \gamma(w) &= \{w,w'\}1 \ , \\
  \gamma(w)^*\! &= \gamma(\rho(w)) \ .
\end{aligned} \label{eq:CARalg}
\end{align}
The CAR-algebra is the $C^*$-algebra generated by the operators $\gamma(w)$ with the relations \ref{eq:CARalg}.
\subsection{States}
A \emph{state} is a linear positive normalized functional $\omega$ on the CAR-algebra. In less technical terms, it is the collection of all correlation functions
\begin{align}
\omega(\gamma(w_1)\cdots \gamma(w_n)) \ .
\end{align}
If a wavefunction $w$ has support in a region $X \subset \mathbb{Z}$, then $\gamma(w)$ is said to be \emph{supported} in $X$. More generally, if $\mathcal{O}$ is a polynomial in $\gamma(w_1),...,\gamma(w_n)$, and $w_1,...,w_n$ all have support contained in $X$, then $\mathcal{O}$ is said to be supported in $X$ and this is indicated by a subscript: $\mathcal{O}_X$. Finally, the translate by $d$ sites of this operator is designated by $\mathcal{O}_{X+d}$.
\\
Importantly, a \emph{pure} state \emph{clusters} \cite{kastler1966invariant,ruelle1966states,miyata1973clustering,bratteli2012operator}; assuming translational invariance:
\begin{align}
\lim_{d\rightarrow \infty} \omega(\mathcal{O}_{1,X} \mathcal{O}_{2,Y+d}) = \omega(\mathcal{O}_{1,X}) \omega(\mathcal{O}_{2,Y}) \ . \label{eq:omegacluster}
\end{align}
In the following it will be assumed that this happens at exponential speed, i.e. there is $C>0$ and $\ell_c < \infty$ such that for all operators $\mathcal{O}_{1,X},\mathcal{O}_{2,Y}$:
\begin{align}
\begin{aligned}
|\omega(\mathcal{O}_{1,X}\mathcal{O}_{2,Y}) - \omega(\mathcal{O}_{1,X})&\omega(\mathcal{O}_{2,Y})| \leq \\
&\leq C \| \mathcal{O}_{1,X}\| \|\mathcal{O}_{2,Y}\| e^{-d(X,Y)/\ell_c} \ .
\end{aligned} \label{eq:omegaexpcluster}
\end{align}
Here, $d(X,Y) := \min\{|x-y| \ : \ x\in X, y \in Y \}$.
\\
Why this restriction? Topological phases are phenomena pertaining to unique gapped ground states of local Hamiltonians. Those are guaranteed\cite{hastings2006,hastings2021gapped} to satisfy condition \ref{eq:omegaexpcluster}. To demand exponential correlations directly simplifies the matter, as the connection between Hamiltonians and their ground states is famously inexplicit. Additionally, Hamiltonians contain vastly more information than just their ground states.
\\
In a first step in the argument, however, there will be a more severe restriction, namely the state $\omega$ will be required to be a \emph{super matrix product state} (sMPS). This notion will receive some further exposition in section \ref{sec:sMPS}, for now the following definition will suffice:
\begin{mydef} \label{def:sMPS}
	A super matrix product state of bond dimension $D \in \mathbb{N}$ is a pure translational invariant state $\omega$ such that the reduced density matrices $\sigma_n(\omega)$ of $\omega$ to $\{1,...,n\}$ satisfy $\textup{rank}(\sigma_n) \leq D$, and $D$ is the smallest such integer.
\end{mydef}
\noindent
This is quite restrictive. Exponentially correlated states satisfy so-called \emph{area laws} for $\alpha > 0$- R\`{e}nyi entanglement entropies \cite{Brand_o_2014}, i.e., such states satisfy
\begin{align*}
\frac{1}{1-\alpha} \log\textup{Tr}(\left[\sigma_n\right]^\alpha) \leq C_\alpha < \infty \ , \quad \alpha > 0 \ .
\end{align*}
If this would hold all the way down to $\alpha = 0$, the assumption in definition \ref{def:sMPS} would include all exponentially correlated states, hence all unique gapped ground states to local Hamiltonians. However, even gapped free-fermion ground states have diverging rank of the reduced density matrices, if their bands are not flat. Therefore, statements made for the set of sMPS have no immediate import to general ground states.
\\
It should be noted that the defining condition of a sMPS of bond dimension $D$ allows to reproduce exactly the expectation values of observables supported on a region of size\footnote{This approximation is much better for exponentially correlated states \cite{Dalzell_2019,schuch2017matrix}.} $\propto \log D$. This allows to construct, for \emph{any} state $\omega$, a series of sMPS $\omega_{\alpha}$, with in general diverging bond dimensions $D_\alpha$, that approximates $\omega$ in the $w^*$-sense \cite{Fannes1992a,Fannes1992b}, i.e., for all observables $\mathcal{O}$ and any $\epsilon>0$ there is an $\alpha_*$ such that $|(\omega-\omega_{\alpha})(\mathcal{O})| \leq \epsilon \|\mathcal{O}\|$ for all $
\alpha > \alpha_*$. In particular, the reduced density matrices converge (in any topology). Therefore generality can be attained -- starting from sMPS -- by considering the limit of diverging bond dimension. Without further constraints, this, however, is again too big a set, as it includes e.g. critical states. To exclude them, here the following set is proposed:
\begin{mydef} \label{def:wellapprox}
	A state that can be written as a $w^*$-limit of sMPS, $\omega_{\alpha} \stackrel{w^*}{\rightarrow} \omega$, where the states $\omega_{\alpha}$ have a uniform bound on their correlation length, will be called a \emph{well-approximable state}. Their set with a given correlation bound $\ell_c$ is denoted by $\mathcal{T}_{\ell_c}$.
\end{mydef}
\noindent
This is a convenient set for it gives analytical control to limits, the limit states are exponentially correlated, and it has been argued that such states have quasi-particle excitations \cite{Zauner_2015,Vanderstraeten_2017} of mass $\propto 1/\ell_c$.

\subsection{Symmetries}
To produce a richer phenomenology assume there is a symmetry action $\mathscr{G}$ acting by homomorphisms on $\mathscr{W}$. This is represented by automorphisms $x \mapsto \alpha_g x (\alpha_g)^{-1}$ satisfying the covariance relation
\begin{align}
\alpha_g \gamma(w) (\alpha_g)^{-1}= \gamma(g^{-1}(w)) \ . \label{eq:gLift}
\end{align}
Here and in the following, assume translational symmetry, and that the rest of $\mathscr{G}$ acts on-site, i.e., $\mathscr{G} = \mathbb{Z} \times G$.
\\
Furthermore, the focus here is on the case that $G$ contains anti-unitary operations $K$. Recalling the picture of $W = V \oplus V^*$, there are the two possibilities, (i) $K : V \rightarrow V$ or (ii) $K: V \rightarrow V^*$. It cannot be more general since otherwise its application would create states of indefinite fermion parity, which is forbidden by a superselection rule. Both of these reverse the direction of time. In fact, they can be brought in a standard form:
\begin{mylem} \label{lem:AntiUstandard}
	Suppose the anti-unitary $K$ generates a group of automorphisms on $W$. Then w.l.o.g. this group factorizes as $U \times A$ where $U$ is an abelian group acting by unitaries, and $A$ is either $\mathbb{Z}_2$ or $\mathbb{Z}_4$, and acts by anti-unitaries.
\end{mylem}
\noindent
A proof can be found in appendix \ref{appendix:StandardForm}.
\\
Assuming that $V$ parametrizes electron orbitals, transformations of the second type, $K:V\rightarrow V^*$, reverse the electric charge, and are therefore called \emph{particle-hole transformation} \cite{Zirnbauer2021}. Note that in this case w.l.o.g. $K^2 =1$. Indeed, if $(K^2)|_V = x$, define an involution $\widetilde{K} = ( \begin{smallmatrix}
0 & K \\
x K & 0
\end{smallmatrix})$. Transformations of this type arise e.g. through sublattice symmetries at half filling.\footnote{For Hubbard models with strong interactions, particle-hole transformations are good symmetries when correlated hopping is negligible \cite{skorenkyy2021electron}.}
\\
The first possibility, $K:V\rightarrow V$, arises, e.g., when the fundamental time-reversal symmetry of electron motion remains unbroken. If such a transformation has $K^2 = -1$, the sign cannot be transformed away. Such an operator is said to be of \emph{time-reversal} type.

\subsection{Homotopies and Invariants}
Restrict the set from defintion \ref{def:wellapprox} to $G$-invariant states and denote that set as $\mathcal{T}_{\ell_c}^G$. In order to filter out all the non-topological properties and simplify the analysis, introduce the concept of a \emph{homotopy of states}. This is done by considering continuous functions $ \nu : [0,1] \rightarrow \mathcal{T}_{\ell_c}^G$. The class of such functions allows to introduce an equivalence relation; $\omega \sim \omega'$ if there is a homotopy of states $\nu$ such that $\nu_0 = \omega$ and $\nu_1 = \omega'$. The main part of this work will be concerned with constructing \emph{homotopy invariants of string-order type} $Z:\mathcal{T}_{\ell_c}^G \rightarrow \mathbb{C}$; these are functions which are constant on all homotopies, and are of the form
\begin{align*}
Z(\omega) := \lim_{n \rightarrow \infty} \omega(\mathcal{O}_n) \ ,
\end{align*}
where $\mathcal{O}_n$ has support $n$. Their precise form is given in section \ref{sec:RP2PartFuncssMPS} below, but it is helpful to take a step back and spend some more time on the super matrix product states introduced in definition \ref{def:sMPS}. This, together with a diagrammatic formalism suitable to incorporate both the fermionic grading and the anti-linear operations, allows then to define, calculate and interpret non-local order parameters for super matrix product states. Afterwards, section \ref{chap:Beyond} will extend the order parameters to $\mathcal{T}_{\ell_c}^G$, using the formalism-independent definition.

\section{Real Super Matrix Product States}
The formalism used is a generalization of translation invariant matrix product states (TI-MPS) \cite{Fannes1992a,10.5555/2011832.2011833,cirac2021matrix} to fermionic systems. These are then called \emph{super} matrix product states (sMPS), since it is most convenient to introduce auxiliary super vector spaces \cite{Bultinck_2017}. There are also other approaches to deal with the fermionic grading, namely Grassmann algebras \cite{Wille2017} or modified tensor contractions \cite{Brugnolo2021}.
\\
The presence of anti-unitary symmetries will force the bond algebra to take a special form, motivating the name `real'.

\subsection{Construction of Super Matrix Product States} \label{sec:sMPS}
To define a states is to construct its correlation functions. This will be done in terms of super vector spaces\cite{10.1007/BFb0082020,moore2014quantum}.
\\
A \emph{super vector space} $V$ is a vector space together with a decomposition $V = V^0 \oplus V^1$. For $\xi \neq 0$ in $V^\mu$ introduce the \emph{parity} $|\xi| = \mu \in \{0,1\}$. Elements of $V^0$ are called \emph{even}, elements of $V^1$, \emph{odd}. All elements with definite parity are called \emph{homogeneous}. A \emph{super Hilbert space} is a super vector space with an inner product such that $V^0$ and $V^1$ are orthogonal. 
\\
For simplicity assume that the CAR-algebra is given with some explicit set of local orbitals $V$. To construct a super matrix product state of bond dimension $D$, choose a $D$-dimensional super Hilbert space $\virtHilb$. Finally, choose an even map $E:\Lambda(V) \rightarrow \mathscr{L}(\virtHilb)$. The image of $E$ in $\mathscr{L}(H)$ is called the \emph{bond algebra}, or the algebra of Krauss operators, and denoted by $A$. This bond algebra is a \emph{super algebra}, i.e., a super vector space with a multiplication that satisfies $A^\mu A^\nu \subseteq A^{\mu+\nu}$.
\begin{mydef} \label{def:sMPSTensor}
	An even map $E:\Lambda(V) \rightarrow \mathscr{L}(\virtHilb)$ is called a \emph{$G$-invariant super matrix product tensor} if the bond algebra satisfies $A^* \subseteq A$ and the map $\grtrans_L \in \mathscr{L}^2(\virtHilb)$,
	\begin{align}
	 \grtrans_L(a) = \sum_{s,t=1}^d (-1)^{|\psi_t||a|} \langle \psi_s,L(\psi_t)\rangle E(\psi_s) a E(\psi_t)^* \ ;\label{eq:grtransEL}
	\end{align}
	where $\psi_{1},...,\psi_d$ is a orthonormal basis of $\Lambda(V)$, satisfies:
	\begin{itemize}
		\item[(i)] $\grtrans_1(e) = e$, with $e$ the identity on $\virtHilb$;
		\item[(ii)] There is a projective $G$-representation $g \mapsto \widehat{\alpha}_g$ s.t.
		\begin{align*}
		(-1)^{|\widehat{\alpha}_g||\psi|}\widehat{\alpha}_g \circ E(\psi) \circ (\widehat{\alpha}_g)^{-1} = E\circ \alpha_{g}(\psi) \ . 
		\end{align*}
		\item[(iii)] There is no even projection $p$ and no integer $n$ such that $(\grtrans_1)^n(p \mathscr{L}(\virtHilb)p) \subseteq p \mathscr{L}(\virtHilb)p$.
	\end{itemize}
\end{mydef}
\noindent
$\grtrans \equiv \grtrans_1$ is often called the \emph{transfer operator}. A glance at its definition reveals that it is a completely positive operator. Therefore, its unitality (i) already implies $\| \grtrans\| = 1$. Condition (iii) for $n=1$, \emph{super-irreducibility}, guarantees for any positive map $\phi$ can be brought into unital form, with the help of a redefinition $\phi \rightarrow z^{-1/2}\phi(z^{1/2}\,\cdot\,z^{1/2})z^{-1/2}$. It furthermore guarantees \cite{https://doi.org/10.1112/jlms/s2-17.2.345} the existence of a strictly positive even linear functional $\lambda$ satisfying $\lambda \circ \grtrans = \lambda$. 
\\
\\
This data determines a functional on operators $L_1,...,L_n$ with supports on sites $1,...,n$ as
\begin{align}
\omega(L_1\cdots L_n) := \lambda \circ \grtrans_{L_1} \circ \cdots \circ \grtrans_{L_n}(e) \ . \label{def:sMPSexpval}
\end{align}
Then define $\omega$ on all other operators by translations. Since $e$ and $\lambda$ are fixed points of $\grtrans$, the expectation value of an operator $\mathcal{O}_{\{1,..,n\}}$ does not change when it is replaces by $1_{\{0\}}\mathcal{O}_{\{1,..,n\}}$ or $\mathcal{O}_{\{1,..,n\}} 1_{\{n+1\}}$.
Furthermore, a short argument will show that $\omega$ is positive. For its formulation, the notion of the \emph{super tensor product}, denoted $V \sotimes W$, is helpful. This is the space
\begin{align*}
 \left[(V^0 \otimes W^0) \oplus (V^1 \otimes W^1)\right] \oplus \left[(V^0 \otimes W^1) \oplus (V^1 \otimes W^0)\right] \ .
\end{align*}
If $V,W$ are super algebras, their tensor product is again a super algebra, with product
\begin{align}
(x_1 \sotimes x_2)\cdot (y_1 \sotimes y_2) := (-1)^{|x_2||y_1|} x_1x_2 \sotimes y_1y_2 \ . \label{eq:supermultiplication}
\end{align}
Now to see that $\omega$ is positive, recall that, by Stinespring's dilation theorem \cite{Stinespring1955,Szehr2016}, completely positive maps $\phi$ are exactly those that can be written as $\phi(a) = U^*\pi(a)U$, where $\pi$ is a $*$-representation and $U$ is an isometry. Thence, introduce isometries $U_n:\virtHilb \rightarrow \Lambda(V)^{\sotimes n} \sotimes \virtHilb $, and denote $E_s \equiv E(\psi_s)$:
\begin{align}
U_n(\xi) := \sum_{s_1,...,s_n} \psi_{s_1} \sotimes \cdots \sotimes \psi_{s_n} \sotimes (E_{s_1})^* \cdots (E_{s_n})^*\xi \ . \label{eq:Undef}
\end{align}
Now observe that for an operator $\mathcal{O}_n$ with support on $\{1,...,n\}$, 
the expectation value is expressed in terms of $U_n$ as $\omega(\mathcal{O})=\lambda((U_n)^*(\mathcal{O} \sotimes e)U_n )$. Hence $\omega$ is positive since $\lambda$ is.
\\
\\
Finally, it is incumbent to check that this construction is in fact compatible with defintion \ref{def:sMPS} given above. The first observation is that $\mathcal{O} \mapsto \grtrans_{\mathcal{O}}(e) = (U_n)^*(\mathcal{O}\sotimes e)U_n$ reduces the rank of an observable, so that $\textup{rank}(\sigma_n(\omega))  \leq \textup{rank}(\lambda) = D$. To see that this $D$ is minimal, again a short digression into the structure of super algebras is helpful. An ideal $I$ in an algebra $A$ is a subalgebra such that $AI \subseteq I \supseteq IA$. $I$ is \emph{graded} if $I = (I\cap A^0) \oplus (I \cap A^1)$.
\begin{mydef} \label{def:supersimple}
	A super algebra is \emph{super-simple} if it has no non-trivial graded ideals.
\end{mydef}
\noindent
For finite dimensional super algebras closed under conjugation, super ideals $I$ are principal ideals, i.e., generated by a central even projection $p$ as $I = pA$. For this reason super-simple algebras are always \emph{super-central}: their even center satisfy $Z(A)^0 = \mathbb{C}e$.
\begin{mylem}
	The bond algebra generated by a super matrix product tensor is super-simple. Its action on $\virtHilb$ is non-degenerate. 
\end{mylem} 
\begin{proof}
	Suppose $I = pA$ was a super ideal in $A$. Then
	\begin{align*}
	\grtrans(pxp) &= \sum_s (-1)^{|E_s||x|} E_s pxp (E_s)^* = \\
	&= p \left[ \sum_s (-1)^{|E_s||x|} E_s x (E_s)^*\right] p = p \grtrans(x)p \ ,
	\end{align*}
	and hence $p$ would reduce $\grtrans$, violating condition (iii) of definition \ref{def:sMPSTensor}. For the second part assume there was a graded subspace of $\virtHilb$ invariant under the action of $A$; the projection $p$ on that subspace is even and reduces $\grtrans$. Note that the two parts are not equivalent since the ideal $pA$ in the second part could be trivial. This would then simply necessitate to update $\virtHilb \rightarrow (1-p)\virtHilb$.
\end{proof}
\noindent
It follows that $\grtrans_{\mathcal{O}}(e)$ generates a super-simple subalgebra of $\mathscr{L}(\virtHilb)$, which, nevertheless, explores fully the even part, i.e., the domain of $\lambda$. This shows that the bond dimension of the state defined by equation \ref{def:sMPSexpval} is indeed $D$. It is also possible to construct more directly\cite{Fannes1992a,Bourne_2021} a super matrix product tensor for any state satisfying condition \ref{def:sMPS}.

\subsection{Properties of Super Matrix Product States} \label{sec:PropertiesofsMPS}
After having constructed super matrix product states, this section derives some of their properties. First, there will be an explanation of how correlation functions are calculated. This leads to an algebraic characterization of the bond algebras of super matrix product states. Afterwards, the modifications required by anti-unitary symmetries are explored.
\\
\\
It still remains to show that the states constructed in equation \ref{def:sMPSexpval} are pure. As was argued before, this is equivalent of clustering of correlations in the sense of equation \ref{eq:omegacluster}. Thus consider regions $X_\pm$ contained in the negative and the positive half space, respectively. Consider operators $\mathcal{O}_{\pm}$ with support in $X_{\pm}$. Their correlations take the form
\begin{align*}
\omega(\mathcal{O}_{-}\mathcal{O}_{+}) = f\circ \grtrans^{d(X_-,X_+)}(q) \ , \quad q = \grtrans_{\mathcal{O}_{+}}(e) \ , \quad f = \lambda \circ \grtrans_{\mathcal{O}_{-}} \ .
\end{align*}
Thus the study of the spectrum of $\grtrans$ is of interest. Using two facts about completely positive maps, it is easy to see that every unit eigenvalue yields a reducing projection. These two facts are (i) $p$ reduces $\grtrans$ if and only if\footnote{This rests on the hereditariness of cones in $C^*$-algebras, see e.g. \cite[Section~1.5]{eilers2018c}.} there is $r > 0$ such that $\grtrans(p) \leq r p$  and (ii) the Kadison-Schwarz inequality\cite{paulsen2002completely} $\grtrans(x^*x) \geq \grtrans(x)^*\grtrans(x)$.
\\
Fact (i) immediately yields that if $p = \grtrans(p)$ is a projection, it reduces $\grtrans$. Next assume $x$ is a fixed point, but not a projection. By the form of $\grtrans$, $x^*$ is also a fixed point, thus w.l.o.g. $x^* = x$. Then by (ii), the $x^n$ are also fixed points, and so is $(z e - x)^{-1}$ for any complex $z$ such that this is defined. Thus, for any eigenvalue $\mu$ of $x$, pick a contour $C_\mu \subset \mathbb{C}$ encircling $\mu$ once counterclockwise and containing no other eigenvalue. Then the spectral projection
\begin{align*}
p_\mu = \oint_{C_\mu} \frac{dz}{2\pi i}  (z e - x)^{-1} 
\end{align*}
is a fixed point of, and thus reduces, $\grtrans$.
\\
Finally assume $\grtrans(x) = u x$ with $|u|=1$. Again, by (ii), $x^n$ is also a fixed point. By finite-dimensionality, there is a $k$ such that $u^k = 1$. Hence $\grtrans^k$ has an eigenvector $x$ to the eigenvalue $1$, hence is reducible.
\\
Introduce the following measure of the largest correlation length in a given system:
\begin{align*}
\textup{Corr}_\omega(k) := \sup\{|\omega(\mathcal{O}_{-}\mathcal{O}_{+}) -\omega(\mathcal{O}_{-})\omega(\mathcal{O}_{+})| \, : \\ \, \|\mathcal{O}_-\| = \|\mathcal{O}_+\| = 1 \ , \ d(X_-,X_+) = k \} \ .
\end{align*}
This leads to the following connection between the transfer matrix spectrum and the correlations in a state:
\begin{myprop} \label{prop:Corrwedge}
	For a pure sMPS $\omega$ with bond algebra $A$ and transfer operator $\grtrans$, there is a constant $1 \geq C>0$ s.t.:
	\begin{align*}
	C \| \grtrans^k - P_e\|_{A} \leq \textup{Corr}_{\omega}(k) \leq \| \grtrans^k - P_e\|_{A} \ , \quad P_e(x) = \lambda(x)e \ .
	\end{align*}
\end{myprop}
\noindent
The proof can be found in appendix \ref{chap:sMPSSpectra}. The second inequality is quite straight-forward and is included for completeness; the first however is less well-known. Both together have the consequence that the correlation length $\ell_c$ of a sMPS and its second largest eigenvalue $\lambda_2$ are precisely related as $|\lambda_2| = \exp(-1/\ell_c)$. 
\\
Proposition \ref{prop:Corrwedge} proves that super matrix product states correlation function cluster, i.e., that the state is pure.
\\
\\
An attentive reading of proposition \ref{prop:Corrwedge} leads to the question of why the operator norms are taken over $A$ only and not $\mathscr{L}(\virtHilb)$. Now $A$ might be not just super-simple but also simple, i.e., satisfy definition \ref{def:supersimple} where the ideals do not need to be graded. In this case $A=\mathscr{L}(H)$. There is the possibility, however, that there is a non-trivial central self-adjoint $\eta \in A^1$ with $\eta^2 = e$, and $A^0$ is a simple algebra \cite{wall1964graded,10.1007/BFb0082020}. Then $A$ contains two ideals $I_{\pm} = p_\pm A$ with $p_{\pm} = (e\pm \eta)/2$. Such ideals are not graded, hence $A$ is still super-simple. It follows that the fermion parity $P$ is not in $A$. Hence there is an additional fixed point of $\grtrans$ in $\mathscr{L}(\virtHilb)\setminus A$:
\begin{align*}
\grtrans(\eta P) &= \sum_s (-1)^{|E_s|} E_s \eta P (E_s)^* = \\
&= \eta P \left[ \sum_s (-1)^{|E_s|} P^{-1}E_sP(E_s)^* \right] = \eta P \grtrans(e) = \eta P \ .
\end{align*}
This additional fixed point $z := i \eta P$ can never appear in correlation functions of local operators. The presence or absence of such a second fixed point is captured in an index $\mu := \dim Z(A)-1$, with $Z(A)$ the center of $A$. It will appear in the non-local order parameters constructed later on. The argument shows that $\mu = 0$ if $A$ is simple and $\mu=1$ otherwise.
\\
\\
This section will conclude with some developments pertaining to sMPS with anti-unitary symmetries. By condition (ii) of definition \ref{def:sMPSTensor}, an anti-unitary symmetry $K$ is lifted to projective representation $\widehat{K}$. It is characterized by two invariants $(\epsilon,\hat{k})$:
\begin{align}
(a) \quad \widehat{K} \widehat{P} \widehat{K}^{-1} = (-1)^{\hat{k}} \widehat{P} \ , \quad \quad (b) \quad  \widehat{K}^2 = e^{i \pi \epsilon} (\pm)^{\widehat{F}} \ . \label{eq:kepsilondef}
\end{align}
Here $\widehat{P} = (-1)^{\widehat{F}}$ is the fermion parity of $\virtHilb$ and in equation (b) the $+(-)$ sign has to be chosen for a particle-hole (time-reversal) symmetry. Here and in the following, the index defined in (a) is always the lower-case version of the symbol used for the transformation it characterizes. So a lifted particle-hole type transformation $\widehat{C}$ has invariant $\hat{c}$, a lifted time-reversal type transformation $\widehat{T}$ has invariant $\hat{t}$. Physical on-site symmetries are always assumed to have trivial invariant, but for formal developments the use of invariants $k,c,t,...$, possibly with subscripts, associated to anti-linear transformations $K,C,T,...$, is convenient.
\\
 If $K \equiv C$ is a particle-hole conjugation, then $\epsilon \in \{0,1\}$. On the other hand, for a time-reversal operation $K \equiv T$,
\begin{align}
(-1)^{\hat{t}} = \widehat{T} \widehat{P} \widehat{T}^{-1} \widehat{P}^{-1} = e^{2\pi i \epsilon} \ .  \label{eq:timereversalepsilonk}
\end{align}
If $\hat{t} = 0$, the time-reversal operation has eigenvectors in the even sector. Otherwise, it has none at all.
\\
\\
Focussing now on the particle-hole case, note that one can choose a $K$-invariant basis $\psi_1,...,\psi_d$ of $\Lambda(V)$. Then the state is completely captured by the real super-simple subalgebra  $A_\mathbb{R} := \textup{Fix}(\widehat{C}) = \{a \in A \, : \, \widehat{C}a\widehat{C}^{-1} = a \}$. These states could therefore be called \emph{real} super matrix product states.
\\
Introduce first a particular easy type of such algebras, the \emph{real division superalgebras}. These are real superalgebras such that all non-zero homogeneous elements are invertible. By the super version of Wedderburn's theorem, $A_\mathbb{R}$ is isomorphic to a matrix algebra $\textup{Mat}_k(\mathbb{D})$ over a real division superalgebra $\mathbb{D}$. Finally, there is the following theorem by Wall-Deligne \cite{geiko2021dyson}:
\begin{table}
	\centering
	\begin{tabular}{r|r|r|l} 
		\toprule
		$\eta$ & $\mathbb{D}$ 	& $(\mu,\hat{c},\epsilon)$ &Realization  \\
		\toprule
		$0$ & $\mathbb{R}$ & $(0,0,0)$ &   \\
		$1$ & $C\ell_1$    & $(0,0,1)$ & $e = \sigma_1$, $P = \sigma_3$, $U = 1$     \\
		$2$ & $C\ell_2$    & $(0,1,0)$ & $e_1 = \sigma_1$, $e_2 = \sigma_3$, $P = \sigma_2$, $U = 1$         \\
		$3$ & $C\ell_3$    & $(1,1,0)$ & $e_a = \sigma_1 \otimes \sigma_a$, $P = \sigma_3 \otimes 1$, $U = \sigma_2 \otimes \sigma_2$   \\
		$4$ & $\mathbb{H}$ & $(0,0,1)$ &          \\
		$5$ & $C\ell_{-3}$ & $(1,0,1)$ & $e_a = i\sigma_2 \otimes \sigma_a$, $P = \sigma_3 \otimes 1$, $U = \sigma_3 \otimes \sigma_2$        \\
		$6$ & $C\ell_{-2}$ & $(0,1,1)$ & $e_1 = i\sigma_2$, $e_2 = i\sigma_3$, $P = \sigma_1$, $U = \sigma_2$         \\
		$7$ & $C\ell_{-1}$ & $(1,1,1)$ & $e = i\sigma_2$, $P = \sigma_3$, $U = \sigma_2$        \\
	\end{tabular}
	\caption{List of the real division superalgebras together with an indication how to realize them in terms of Pauli matrices. Here $C = U\widehat{C}_0$ where $\widehat{C}_0$ is elementwise complex conjugation.} \label{table:superDivision}
\end{table}
\begin{mythm}
	There are ten real associative division superalgebras:
	\begin{align*}
	\mathbb{C}, \ \mathbb{C}\ell_1; \ \mathbb{R},  \ C\ell_{1},\ C\ell_{ 2}, \ C\ell_{3}, \ \mathbb{H} \ , \ C\ell_{-3},\ C\ell_{- 2}, \ C\ell_{- 3} 
	\end{align*}
\end{mythm}
\noindent
$\mathbb{R},\mathbb{C},\mathbb{H}$ are the the classical real division algebras found by Frobenius, then there is the complex Clifford algebra with one generator, and the real Clifford algebras with generators $(e_a)^2=\pm 1$. Note also that $\mathbb{C}$ and $\mathbb{C}\ell_1$ are not central as algebras over $\mathbb{R}$, so they can be excluded from this list in the following. The data $(\mu,\hat{c},\epsilon)$ is in $1$:$1$-correspondence with the remaining $8$ division superalgebras \cite{Bultinck_2017}. Here it is convenient to introduce a $\mathbb{Z}/8\mathbb{Z}$ index $\eta = \mu  + 2 \hat{c} + 4 \epsilon$. In table \ref{table:superDivision} these are listed, together with possible realizations of these algebras in terms of Pauli and Dirac matrices. For a further understanding recall that $C\ell_m \sotimes_\mathbb{R} C\ell_n \cong C\ell_{n+m}$, $C\ell_4 \cong \mathbb{H}$ and that $C\ell_{m+8}$ is a matrix algebra over $C\ell_m$, so that $C\ell_{\eta}$ is a matrix algebra over the Division ring with index $\eta$.
\\
This thus establishes that an sMPS with a particle-hole symmetry is characterized by an algebraic invariant $\eta$ that in turn determines a real division superalgebra $\mathbb{D}$ such that the matrices defining the state $E_s$ can be chosen from $\textup{Mat}_k(\mathbb{D})$.
\\
\\
This has physical consequences, e.g., in the entanglement spectrum, controlled by $\lambda =: \textup{tr}(\rho\, \cdot)$. Then $\rho \in \textup{Mat}_k(\mathbb{D})$, and eigenspaces of $\rho$ are modules over the complexification $\mathbb{D}_\mathbb{C}:= \mathbb{D} \otimes_\mathbb{R} \mathbb{C}$. Since $\dim_\mathbb{C} \mathbb{D}_\mathbb{C} \neq 1$ for $\mathbb{D} \neq \mathbb{R}$, the spectrum is degenerate \cite{fidkowski2010entanglement}. 
\\
Now turn to the time-reversal operation. In that case it is advantageous to consider the even part $ A^0$ only. This yields a semi-simple real algebra $(A^0)_\mathbb{R} :=\textup{Fix}(\widehat{T}) $. In fact, $A^0 = B \oplus B$, for some simple real algebra $B$ \cite[Lemma~6]{wall1964graded}. Hence $B$ is a matrix algebra over $\mathbb{D} \in \{\mathbb{R},\mathbb{H}\}$. 
\\
For more general symmetry groups $G$ containing unitary transformations $G_0$ and anti-unitaries, the bond space decomposes into irreducible representations of $G$. An entanglement operator $\rho$, which has to preserve this structure, will then be a block matrix, where the anti-unitary symmetries affect the type of the blocks that can occur \cite{geiko2021dyson}.
\\
Here, however, the focus will remain on the invariants $\eta, \hat{t}$ which characterize broader features of the bond algebra. Section \ref{sec:RP2PartFuncssMPS} will define certain non-local order parameters in terms of a given state $\omega$ that can extract the invariants, following suggestions from the literature. Afterwards, part \ref{chap:Beyond} will generalize this to the well-approximable states of definition \ref{def:wellapprox}. First, however, section \ref{sec:Diagrammatics} introduces a diagrammatic formalism that is useful in calculations.

\subsection{Diagrammatics} \label{sec:Diagrammatics}
The calculations are most conveniently done in a diagrammatic formalism that will be shortly introduced here, with a special focus also on the implementation of anti-unitary operations.
\\
Given a super vector space $V$, an element $\xi \in V$ is simply represented by a node labeled by the vector and an outgoing arrow:
\begin{align*}
\xi =	\raisebox{-1.225ex}{\includegraphics{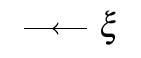}}  .
\end{align*}
Similarly, an element of the dual space $\varphi \in V^*$ is represented by a node labeled by the covector, with an ingoing arrow, 
\begin{align*}
\varphi = \raisebox{-1.225ex}{\includegraphics{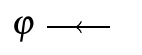}}  .
\end{align*}
The notation is chosen to suggest contractions by attaching the outgoing line of a vector to the ingoing line of a covector, 
\begin{align*}
\varphi(\xi) = \raisebox{-1.225ex}{\includegraphics{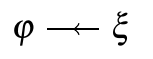} }.
\end{align*}
Finally, linear operators are indicated by boxes:
\begin{align*}
\raisebox{-1.1ex}{\includegraphics{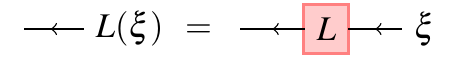}} .
\end{align*}
Super tensor products are implemented by putting all the nodes in a horizontal line, with the convention that exchange of nodes introduces a Koszul sign:
\begin{align}
\raisebox{-1ex}{\includegraphics{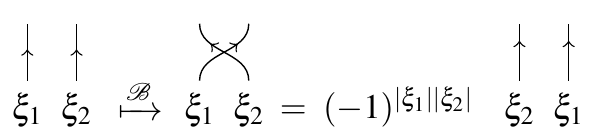}} \ .\label{eq:Braiding}
\end{align}
The matrix product tensor and its adjoint $E^*(\langle \psi, \cdot \rangle) = E(\psi)^*$ are depicted as
\begin{align*}
&\includegraphics{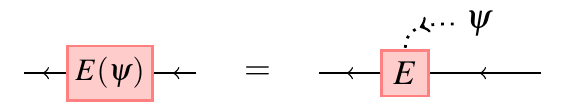}
	\\
	\ \  \ &\includegraphics{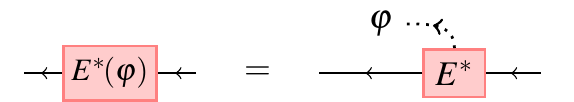} 
\end{align*}
\noindent
\begin{wrapfigure}{l}{0.25\textwidth}
	\vspace{7pt}
	\centering 
	\includegraphics[width=0.25\textwidth]{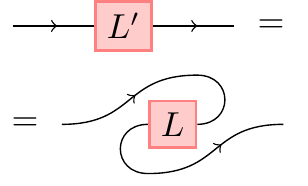}
	\vspace{-20pt}
\end{wrapfigure}
\begin{mydef} \label{ref:gradeddual}
	Let	$L:V\rightarrow W$ be a linear map between super vector spaces $V$ and $W$. Then, the \emph{dual} to $L$ is the linear operator $L':W^* \rightarrow V^*$ defined by
	\begin{align}
	L'(\varphi) = (-1)^{|L||\varphi|} \varphi \circ L  \label{eq:Lgradeddual} 
	\end{align}
\end{mydef}
\noindent
Then introduce $\overline{E}(\varphi):= E^*(\varphi)^\prime$. The operators $\grtrans_L$ of equation \ref{eq:grtransEL} are presented diagrammatically as
\begin{align*}
	\raisebox{-10ex}{\includegraphics{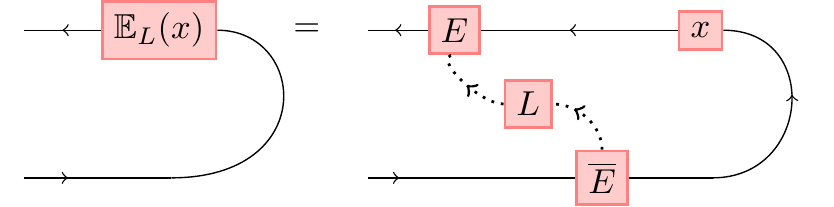}} \ .
\end{align*}
The $\lambda$-functional is presented in terms of $ \Lambda = \widehat{P}\rho_E$ as
\begin{align*}
	\includegraphics{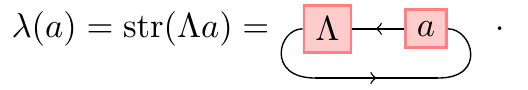}
\end{align*}
The expectation value functional defined in equation \ref{def:sMPSexpval} takes the form:
\begin{align}
\hspace{-2.5ex}	\raisebox{-7.5ex}{\scalebox{0.7}{\includegraphics{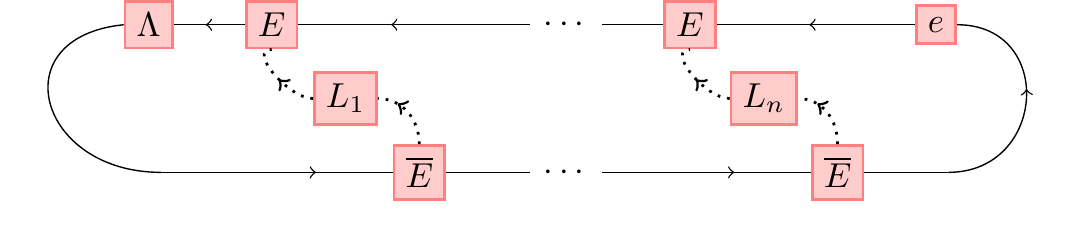} }}  \hspace{-2ex}   \label{eq:omegadiagram}
\end{align}
Diagrammatically representing anti-unitary symmetries is problematic as the diagrams are linear. Therefore the Hilbert structure is used to produce bilinear forms. It requires modifications to be compatible with the grading.
\begin{mydef}
	Suppose $V$ is a super vector space. A non-degenerate even sesquilinear form $h$ is called a \emph{super hermitian form} if it satisfies
	\begin{align}
	\overline{h(\xi_1,\xi_2)} = (-1)^{|\xi_1||\xi_2|} h(\xi_2,\xi_1) \ . \label{eq:gradedsesqui}
	\end{align}
\end{mydef}
\noindent
A super hermitian form $h$ and a super bilinear form on a super Hilbert space $(V,\langle \cdot, \cdot \rangle,K)$ are defined as\footnote{The exponent is chosen as $|\xi|^2$ instead of $|\xi|$ for two reasons: (a) The degree of a vector is a $\mathbb{Z}_2$ variable. Hence, $i^{|\xi|}$ is ill-defined. Consider, e.g., $i^{|L(\xi)|} = (-1)^{|L||\xi|} i^{|L|}i^{|\xi|} \neq i^{|L|}i^{|\xi|}$. This ambiguity is removed by using $|\xi|^2$. (b) More formally, $Q: \mathbb{Z}_2 \rightarrow \mathbb{Z}_4,\ x\mapsto x^2$ is a \emph{quadratic refinement} of the bilinear braiding pairing $B: \mathbb{Z}_2 \times \mathbb{Z}_2 \rightarrow \mathbb{Z}_2, \ (x,y) \mapsto x y$, i.e., $Q$ and $B$ satisfy the relation $2 \cdot B(x,y) = Q(x+y) - Q(x) - Q(y)$, where $2\cdot :\mathbb{Z}_2 \rightarrow \mathbb{Z}_4, \ x \mapsto 2x$. This generalizes to other graduations than $\mathbb{Z}_2$.}
\begin{align}
h(\xi_1,\xi_2) := i^{|\xi_1|^2} \langle \xi_1,\xi_2\rangle \ , \quad \kappa(\xi_1,\xi_2) := h(K(\xi_1),\xi_2) \label{eq:kappadef}
\end{align}
This definition is advantageous as it factorizes nicely under super tensor products; let $\kappa_1,\kappa_2,\kappa_{12}$ be bilinear forms defined by anti-linear $K_1,K_2,K_1 \sotimes K_2$ and inner products $\langle \cdot,\cdot \rangle_1,\langle \cdot,\cdot \rangle_2,\langle\cdot,\cdot\rangle_{12}$ on vector spaces $V_1,V_2,V_1\sotimes V_2$, and
\begin{align}
\begin{aligned}
&[\kappa_1 \sotimes \kappa_2](v_1 \sotimes v_2 , w_1 \sotimes w_2) := \\
 & \quad \quad =(-1)^{|v_2||w_1|+k_2(|v_1|+|w_1|)} \kappa_1(v_1,w_1) \kappa_2(v_2,w_2) \ . \end{aligned}\label{eq:kappa1tenskappa2}
\end{align}
Then $\kappa_{12} = \kappa_1 \sotimes \kappa_{2}$.
\\
Consequently, such bilinear forms have a neat diagrammatic representation:
\begin{align}
\raisebox{-1.0ex}{\includegraphics{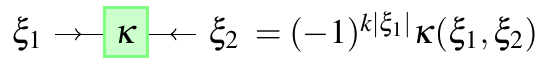} } \ .
\end{align}
A bilinear form defines a transpose:
\begin{align}
\kappa(L^\textit{t} \xi_1,\xi_2) : =
(-1)^{|L|| \xi_1|} \kappa(\xi_1,L\xi_2) \ . \label{eq:kappabilinKtrans}
\end{align}
This is conveniently expressed diagrammatically:
\begin{align}
\raisebox{-1.25 ex}{\includegraphics{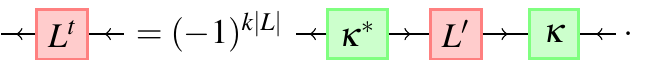}} \label{eq:Ktranspose}
\end{align}
Since the anti-unitaries used in this work have a special form, symmetry properties of $\kappa$ can be derived, depending on the symmetry properties of the underlying $K$. These relations are equations \ref{eq:tausym} and \ref{eq:chisym}. In the following, denote the bilinear forms induced by $T$ and $C$ as $\tau$ and $\chi$, respectively. The two invariants $(k,\epsilon)$ introduced in equation \ref{eq:kepsilondef} are conveniently grouped as $q := k + 2\epsilon$. For a time-reversal operation, equation \ref{eq:timereversalepsilonk} shows that $q = 2 t$.
\begin{align}
&\scalebox{1}{\raisebox{-1.2ex}{\includegraphics{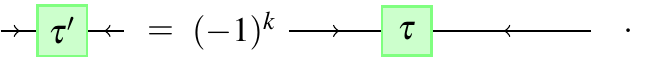}}} \label{eq:tausym} \\
& \scalebox{1}{\raisebox{-1.2ex}{\includegraphics{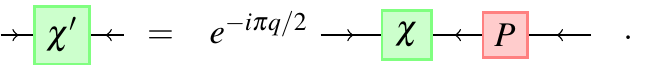}}} \label{eq:chisym}
\end{align}
Finally, return to the setting of a matrix product tensor $E$ satisfying a covariance condition in terms of a physical anti-linear operation $K$ and a lifted $\widehat{K}$. Both can be used to introduce graded bilinear forms $\kappa$ and $\widehat{\kappa}$, respectively.
Then
\begin{align}
E(\kappa^*\varphi) = E^*(\varphi)^\grtranspose \ . \label{eq:EkappaTranspose}
\end{align}
Or, in diagrams:
\begin{align}
\raisebox{-8ex}{\includegraphics{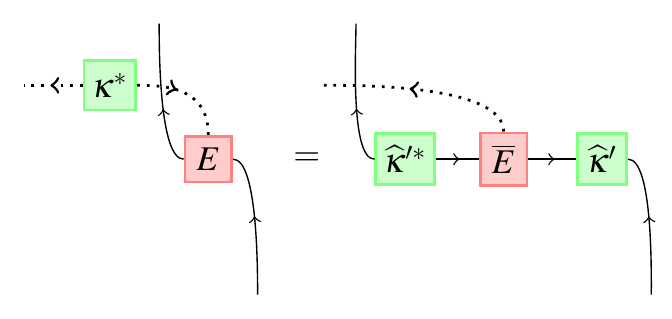}} \label{eq:Ekappapush} 
\ , \\  \raisebox{-8ex}{\includegraphics{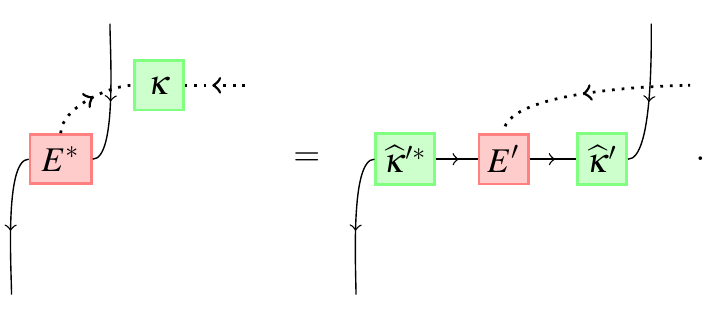}}\label{eq:Estarkappapush}
\end{align}

\subsection{Homotopy Invariants of sMPS} \label{sec:RP2PartFuncssMPS}
By definition topological phases are not characterized by local order parameters, which has shifted the attention to non-local ones \cite{PerezGarcia2008,pollmann2017symmetry}. Most known are those for spin chains with unitary symmetries, which are illustrated in figure \ref{fig:stringorderparameter}.
\begin{figure}
	\centering
	\includegraphics[scale=0.7]{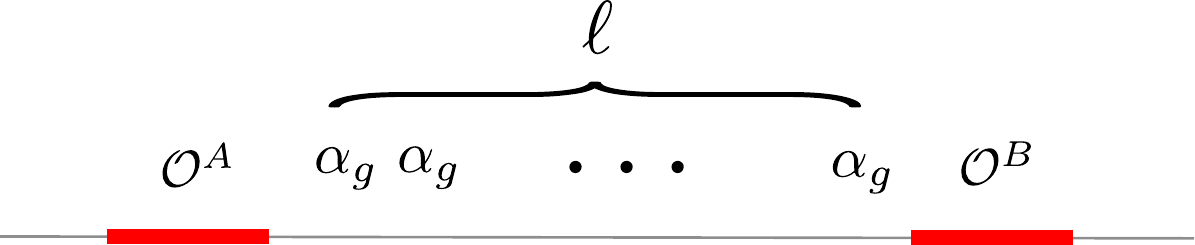}
	\caption{Illustration of the operator setup to compute string order parameters $Z(g):= \lim_{\ell \rightarrow \infty} \omega\left(\mathcal{O}^A \sotimes (\alpha_g)^{\sotimes \ell} \sotimes \mathcal{O}^B \right)$. The operators $\mathcal{O}^{A,B}$ have to be engineered depending on the order that is to be detected.}
	\label{fig:stringorderparameter}
\end{figure}
Notice that these become ambiguous for anti-unitary symmetries, just as the overlap $\langle \Psi_1,T(\Psi_2)\rangle$ is not invariant under redefinitions $\Psi_i \rightarrow e^{i\varphi} \Psi_i$. Information about the relative phase of these rays can be extracted\cite{bargmann1964note} by the function $\langle \Psi_1,T(\Psi_2)\rangle \langle T(\Psi_2),\Psi_3\rangle \langle \Psi_3,\Psi_1\rangle$. Returning to string-order parameters, the step from simple overlaps to the invariant three-vector expression is similar to going from expectation values to expression quadratic in the density operator. Such invariants have been proposed for spin chains with time-reversal symmetry, and calculated for matrix product states\cite{Pollmann2012Detec}. The idea was taken up in a series of papers of Shiozaki et al. and generalized to fermion systems \cite{Shiozaki2017,Ryu2017,Ryu2018,Shapourian2017}. They proposed a formulation in terms of partial transposes, which builds up from the notion of transpose of equation \ref{eq:Ktranspose}:
\begin{mydef} \label{def:partialtranspose}
	If $A_1,A_2$ are superalgebras with graded transposition,
	\begin{align}
	x_1 \sotimes x_2 \mapsto (x_1 \sotimes x_2)^\textup{t} =   (-1)^{k(|x_2|+|x_1|)} (x_1)^{\textup{t}} \sotimes (x_2)^{\textup{t}} \ ,
	\end{align}
	introduce the \emph{partial graded transposes} :
	\begin{align}
	\begin{aligned}
	(x_1 \sotimes x_2)^{\textup{t}_1} &:= (-1)^{k|x_2|}  (x_1)^{\textup{t}} \sotimes x_2 \ , \\
	\quad (x_1 \sotimes x_2)^{\textup{t}_2} &:= (-1)^{k|x_1|} x_1 \sotimes (x_2)^{\textup{t}} \ .
	\end{aligned}
	\end{align}
\end{mydef}
\noindent
Denote the partial graded transpose on a subset $X \subset \mathbb{Z}$ as $\mathcal{O}^{\grtranspose_X}$. If the subset $X$ is of the form $\{1,...,n\}$ then the partial transpose will be abbreviated as $\mathcal{O}^{\grtranspose_{\{1,...,n\}}} \equiv \mathcal{O}^{\grtranspose_n}$.
Next, for
\begin{align*}
{\{1,...,k_1,\} \cup \{d+k_1+1,...,d+k_1+k_2\}} \ ,
\end{align*}
denote the reduced density matrix of a state $\omega$ as $ \sigma_{k_1,k_2|d}(\omega)$. Then, define:
\begin{align}
Z_{k,\ell}^{C}(\omega) & := \Tr(\sigma_{k+\ell} \left[\sigma_{k+\ell}\right]^{\grtranspose_k}) \ , \label{eq:ZklC}  \\
Z_{k_1,k_2|d}^T(\omega) &:= \textup{Tr}(\sigma_{k_1,k_2|d} [\sigma_{k_1,k_2|d}]^{\textup{t}_{k_1}}) \ ,  \label{eq:Zk1k2T}\\
\Delta_n(\omega) & := -\log\Tr(\left[\sigma_n\right]^2) \ .
\end{align}
The last line is the second Renyi entropy, which will play a distinguished r\^{o}le.
\\
There is another way of understanding these invariants. Starting from the working assumption that the low-energy behavior of topological phases is described by a topological field theory, the latter can be used to classify the former \cite{Kapustin2015}. Recent advances \cite{Yonekura:2018ufj,freed2019reflection} show that to determine a given topological field theory\footnote{There are two additional assumptions necessary, (i) unitarity, which comes for free for the physicist; and (ii) invertibility, which models the absence of topological order.}, it suffices to know its partition function on a relatively small number on manifolds, who have also been determined for a wide variety of symmetry groups \cite{kirby1990p,Ryu2018}. The objects \ref{eq:ZklC}, \ref{eq:Zk1k2T} can then be understood as discretizations of a partition function on the real projective plane and the Klein bottle, respectively. This was the argument with which Shiozaki et al. argued that these would be good order parameters. Here, this motivation from topological field theory will not be crucial, instead the line of argument will be to show directly that they extract the algebraic invariants of super matrix product states introduced earlier.

\begin{myprop} \label{prop:ZComega}
	Let $\omega$ be a pure super matrix product state with bond algebra $A$ and virtual parity operator $\widehat{P}$. Then $\Delta(\omega) := \lim_{n\rightarrow \infty} \Delta_n(\omega)$ exists and 
	\begin{align}
	|e^{-\Delta_n(\omega)} - e^{-\Delta(\omega)}| &\leq 2 \| \grtrans^n - P_1\| \ . \label{eq:Deltanconvergence}
	\end{align}
			Here, $P_1$ is the projection on the eigenvalue $1$ in $\mathscr{L}(\virtHilb)$.
	Moreover:
	\begin{itemize}
		\item [(i)] Suppose $\omega$ is invariant under a particle-hole symmetry $C$. Denote by $\widehat{C}$ the lift of $C$ on the bond space, and recall the three indices $\mu$, $(-1)^{\hat{c}} = \widehat{C}\widehat{P}\widehat{C}^{-1}\widehat{P}$ and $(-1)^{\epsilon} = \widehat{C}^2$. Then $Z^C(\omega) := \lim_{k,\ell \rightarrow \infty} Z_{k,\ell}^C(\omega)$ exists and 
		\begin{align}
		\begin{aligned}
		Z^C &= \exp\!\left(-\frac{3}{2} \Delta -2\pi i \frac{\eta_C}{8}   \right) \ , \\
		\eta_C(\omega) &=  4 \epsilon + 2 \hat{c} +\mu \in \mathbb{Z}/8\mathbb{Z} \ .
		\end{aligned}
		\end{align}		
		The convergence is estimated by
		\begin{align}
		|Z^C_{k,\ell}(\omega) - Z^C(\omega)| &\leq 2\left( \| \grtrans^k - P_1\| +  \| \grtrans^\ell - P_1\|\right) \ . \label{eq:ZklConvergence}
		\end{align}
		\item[(ii)] Suppose $\omega$ is invariant under a time-reversal symmetry $T$ with lift $\widehat{T}$. Let $(-1)^{\hat{t}}=\widehat{T}\widehat{P}\widehat{T}^{-1}\widehat{P}$. Then $Z^T(\omega) := \lim_{k_1,k_2,d \rightarrow \infty} Z_{k_1,k_2|d}^T(\omega)$ exists and
		\begin{align}
		\begin{aligned}
		Z^T &= \exp\!\left(-2  \Delta + 2\pi i \frac{\eta_T}{2}  \right) \ , \\ \eta_T(\omega) &=  \hat{t}  \in \mathbb{Z}/2\mathbb{Z} \ ,
		\end{aligned}
		\end{align}
		The convergence is estimated by
		\begin{align}
		\begin{aligned}
		|Z^T_{k_1,k_2|d}&(\omega) - Z^T(\omega)|  \leq \\
		&\leq 2\left( \| \grtrans^{k_1} - P_1\| + \| \grtrans^{d} - P_1\| +  \| \grtrans^{k_2} - P_1\|\right) \ .
		\end{aligned} \label{eq:ZTklConvergence}
		\end{align}
	\end{itemize}
\end{myprop}
\noindent
The calculation of $Z^C$ will be presented with some level of details, while $Z^T$ and $\Delta$, whose calculation follows similar steps, will be treated rather shortly.
\\
In order to obtain a diagrammatic expression for $Z^C_{k,\ell} \equiv Z_{k,\ell} = \Tr(\sigma_{k+\ell} [\sigma_{k+\ell}]^{\grtranspose_k}) $, obtain the density operator from equation \ref{eq:omegadiagram}. Taking the trace -- instead of a super trace -- necessitates the inclusion of a string of fermion parity operators into the diagram. Finally, the partial transpose can be shifted to the virtual level by the covariance equations \ref{eq:Ekappapush} and \ref{eq:Estarkappapush}. The result then takes the form
\begin{align}
\hspace{-5ex}
\scalebox{0.5}{\raisebox{-20ex}{\includegraphics{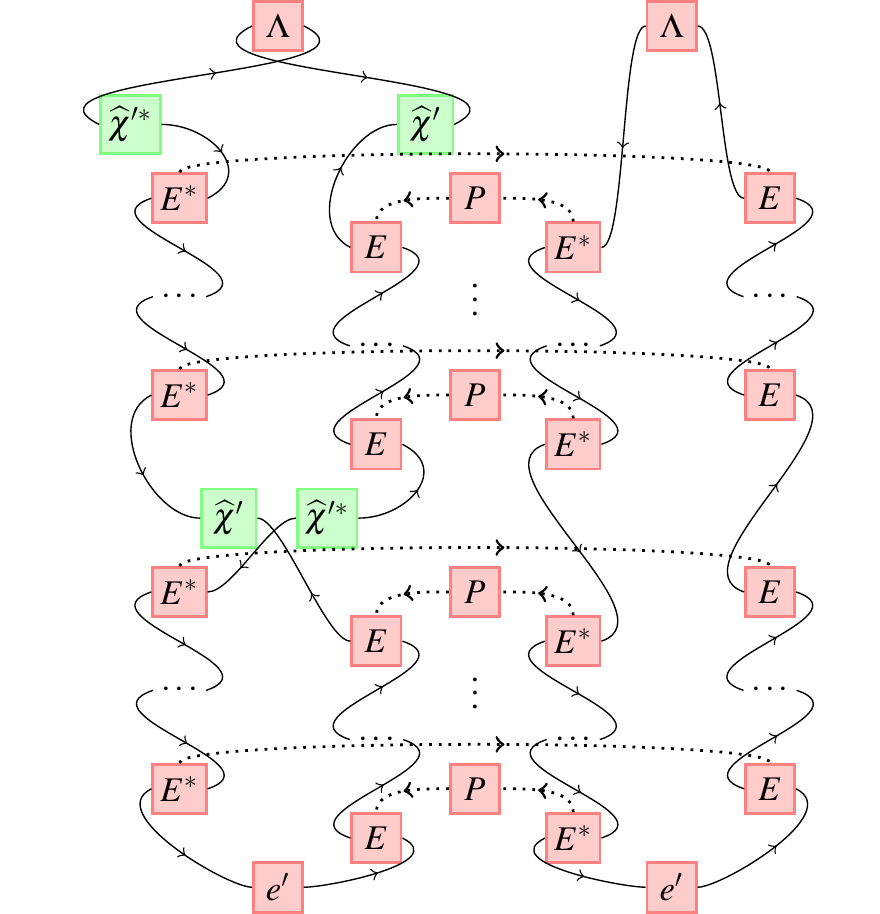}}  = 	\hspace{-8ex}\raisebox{-20ex}{\includegraphics{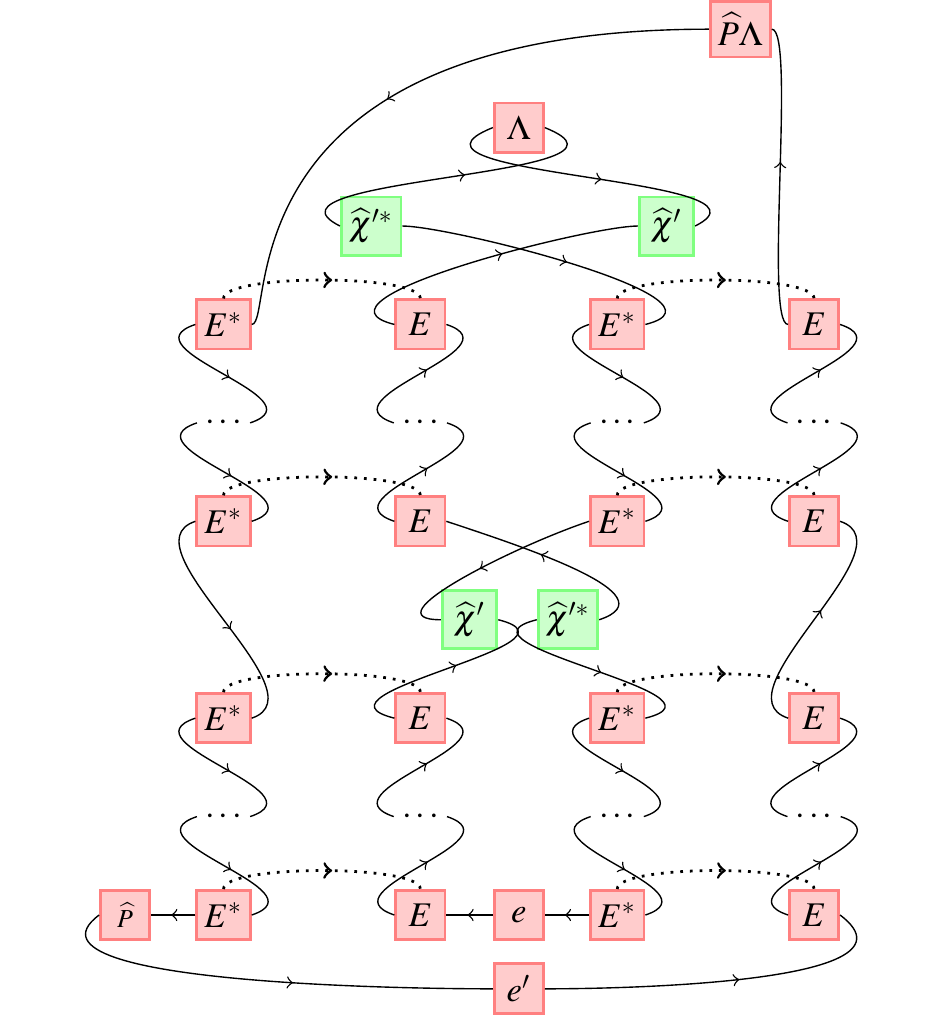}}  } \hspace{-3ex} \label{eq:Zkl}
\end{align}
The limits are taken with the help of
\begin{align}
\lim_{k\rightarrow \infty} \grtrans^k = \lim_{k \rightarrow \infty} \hspace{-2ex} \scalebox{0.7}{\raisebox{-15.5ex}{\includegraphics{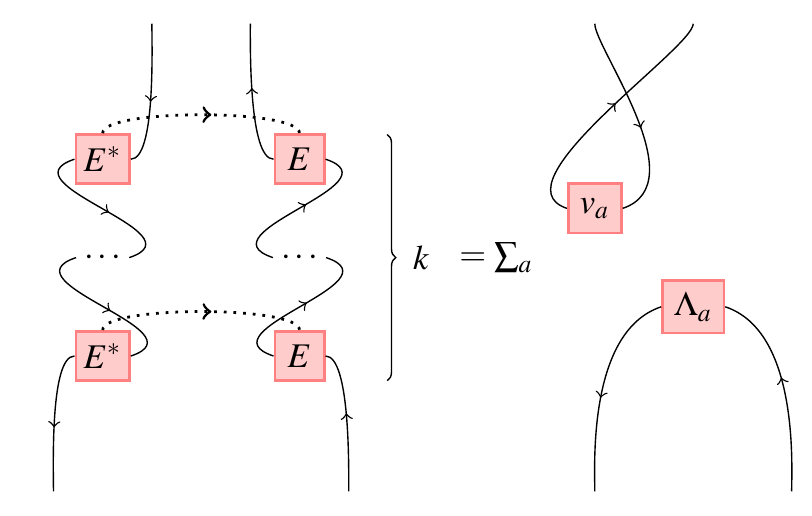}}} \ , \label{eq:transElimitVertical}
\end{align}
where $ v_0 = e,\  v_1 = z, \ \Lambda_0 = P\rho ,\  \Lambda_1 = P\rho z $. Recall that $z = i \eta P$ is the second fixed point of $\grtrans$ introduced below proposition \ref{prop:Corrwedge} and is only present if the bond algebra is not simple ($\mu = 1$ case).
\\
Consequently $Z^C = \sum_{abcd} \textup{I}^{\phantom a}_{ab} \textup{I\!I}^{ab}_{cd} \textup{I\!I\!I}_{\phantom a}^{cd}$ with
\begin{align}
\raisebox{-20ex}{\includegraphics[scale=0.5]{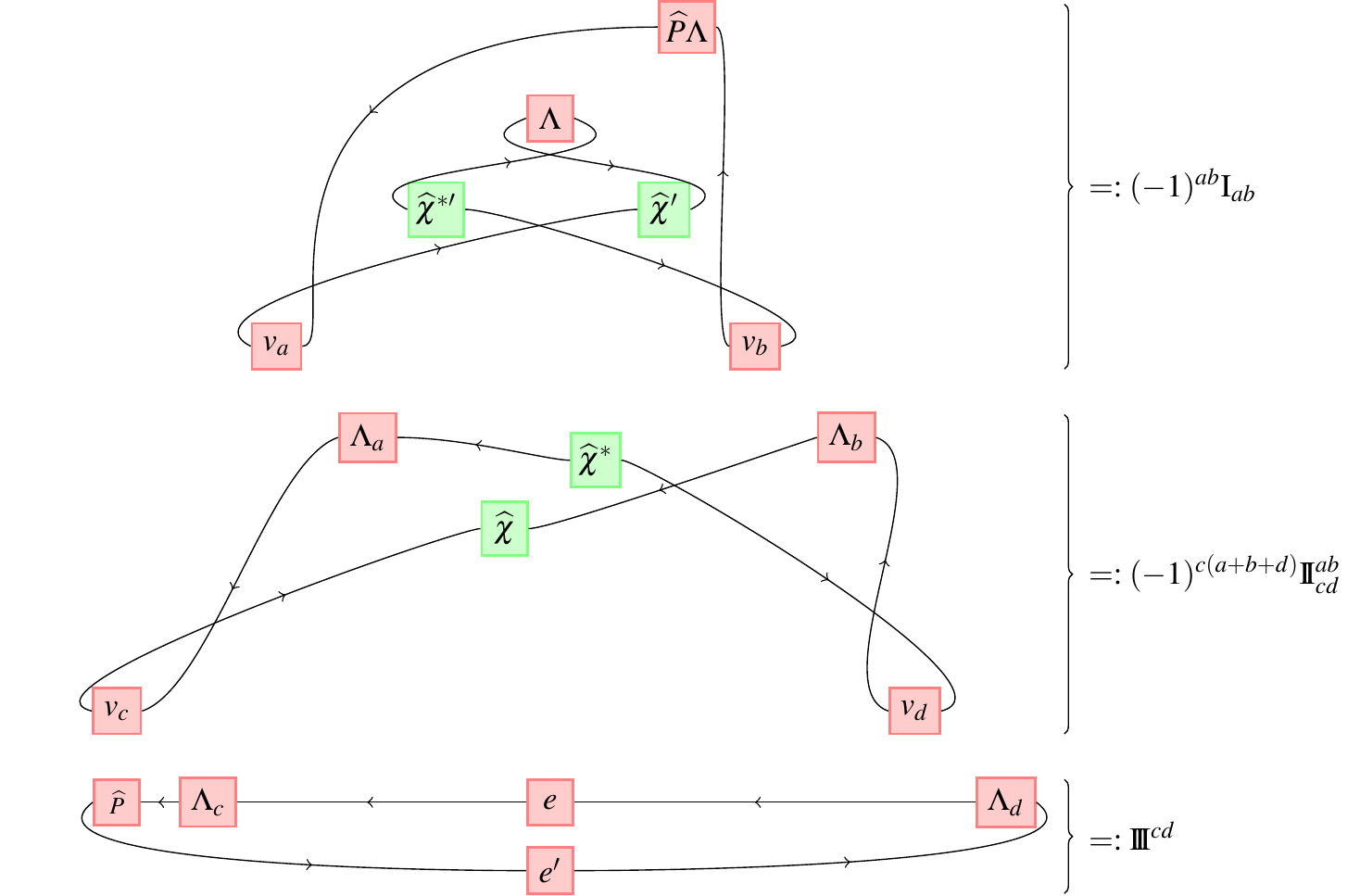}} \label{eq:ZCsubdia}
\end{align}
Separating the subdiagrams (with signs) gives the following expressions:
\begin{align}
\textup{I}_{ab} &= (-1)^{ab} \textup{str}(\widehat{P}\Lambda v_b \Lambda^\textup{t} v_a) = (-1)^{\hat{c}} \delta_{ab}\textup{tr}( \rho^2) \ ,  \\
\textup{I\!I}^{aa}_{bb} &= (-i)^{q} \textup{str}(v_b \Lambda_a (\widehat{P} \Lambda_a v_b)^\textup{t} ) = \nonumber \\
&= (-i)^{q}(-1)^{\hat{c}} (-i)^{(a-b)^2} \textup{tr}(\rho^2) \ ,\\
\textup{I\!I\!I}^{cd} &= \delta_{cd} \textup{tr}( \rho^2) \ .
\end{align}
They can be combined to give
\begin{align*}
Z^C(\omega) &=
\begin{cases}
(-i)^{q}  \left[\textup{tr}( \rho^2)\right]^3 &\ ,\quad \mu = 0  \\
\left[\  (-i)^{q} \  \frac{1 -  i}{\sqrt{2}} \right] \left[\sqrt{2}\textup{tr}( \rho^2)\right]^3 &\ , \quad \mu = 1  \\
\end{cases}
\\
&= \left[ 2^\mu \textup{tr}(\rho^2)^2\right]^{\frac{3}{2}} \exp\!\left(-\frac{2\pi i}{8}( 4 \epsilon + 2\hat{c} + \mu )\right) \ .
\end{align*}
A similar translation procedure for $Z^T$ gives for the limit
$k_1,k_2,d \rightarrow \infty$:
\begin{align}
\hspace{-8ex}\scalebox{0.6}{ \raisebox{-40ex}{ \includegraphics{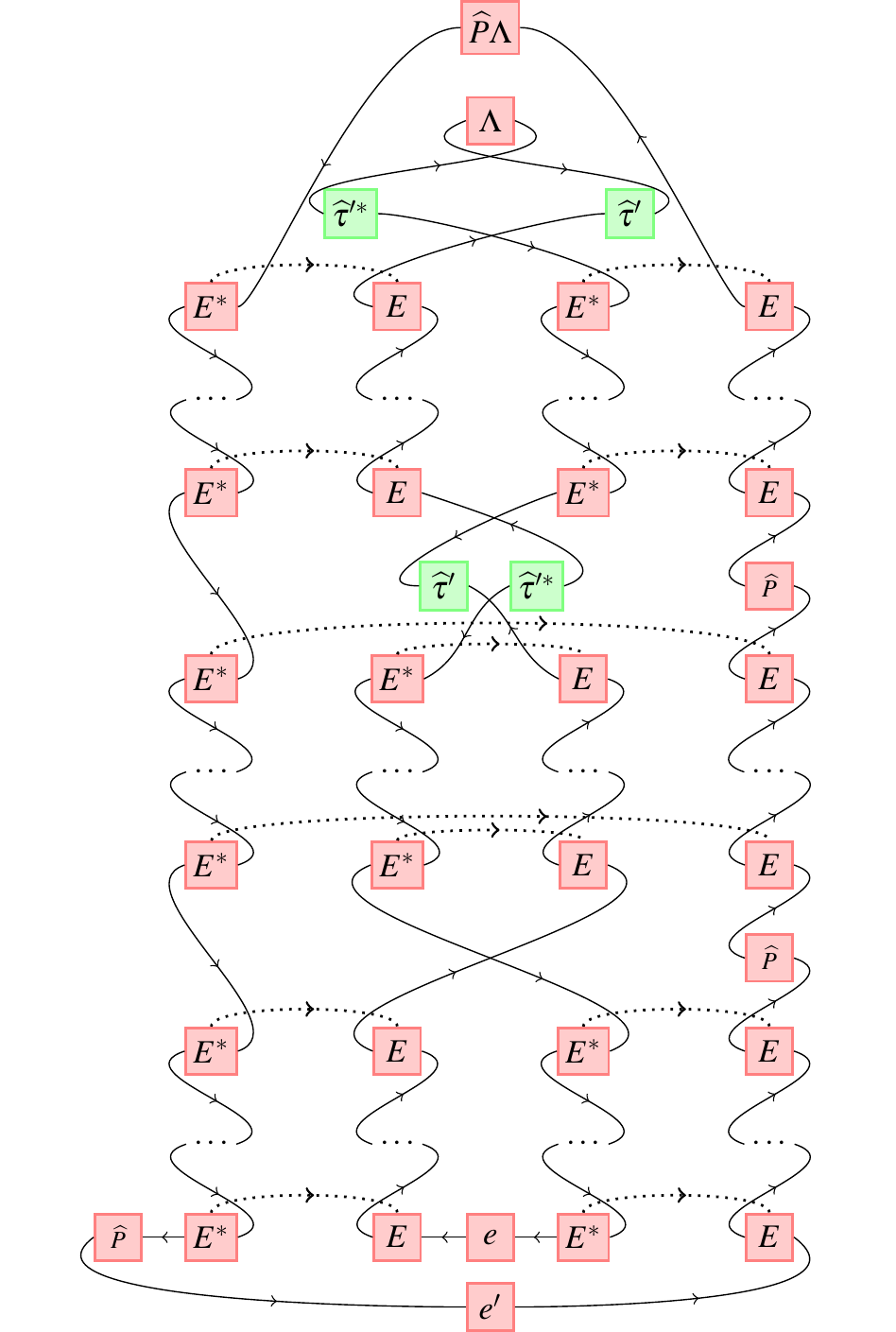} }  \hspace{-4ex} $\rightarrow$ \hspace{-8ex} \raisebox{-40ex}{ \includegraphics{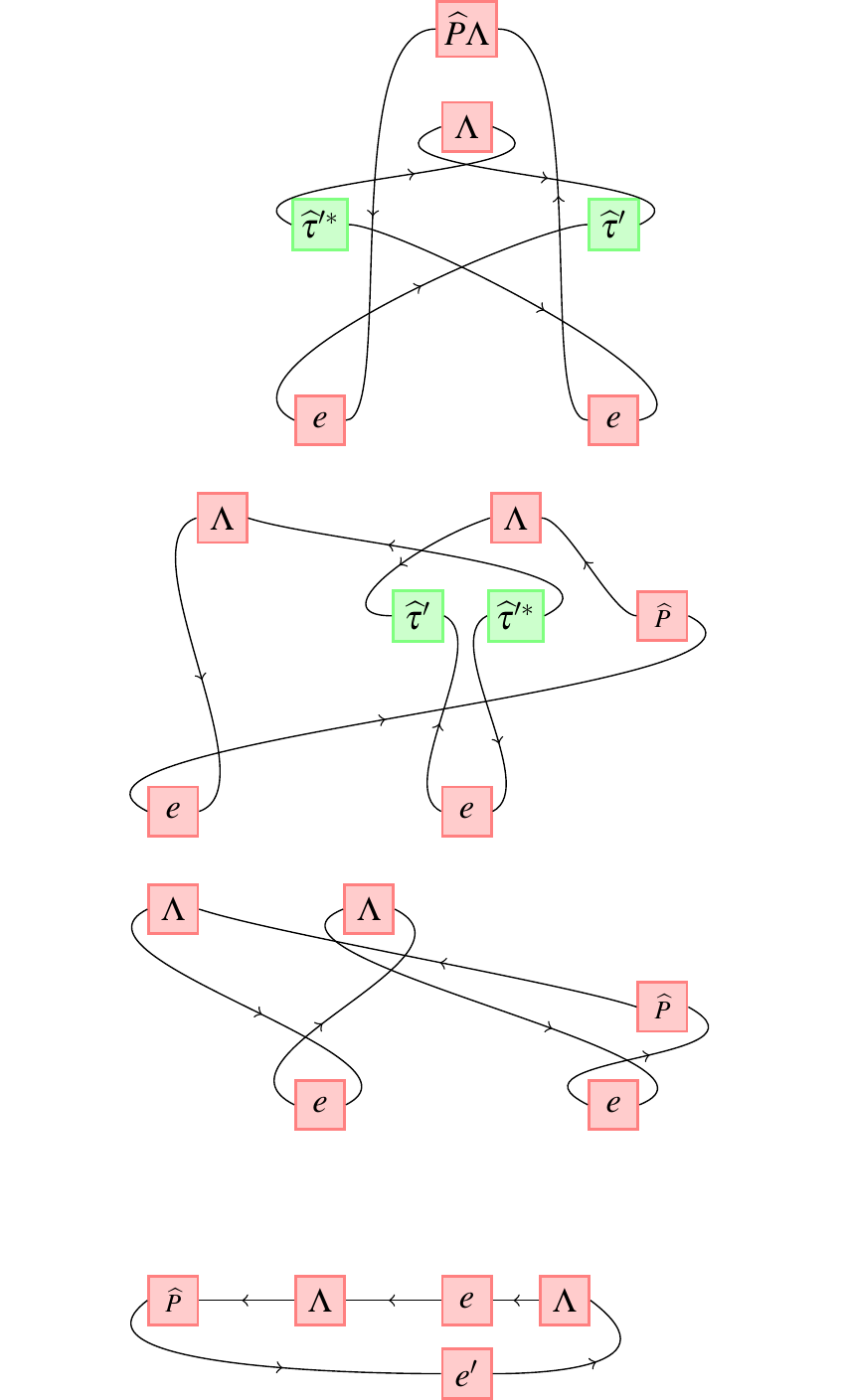}} }\hspace{-10ex} \label{eq:Zkkdlimit}
\end{align}
The uppermost and the lowermost subdiagrams correspond to subdiagrams $\textup{I}$ and $\textup{III}$, respectively, of \ref{eq:ZCsubdia}. The other two are new, but both just give another factor of $\textup{tr}(\rho^2)$. Hence:
\begin{align}
Z^T(\omega) = (-1)^{\hat{t}} \textup{tr}(\rho^2)^4 \ .
\end{align}
The last ingredient, $\Delta_n(\omega)$, is treated quite analogously:
\begin{align*}
e^{-\Delta(\omega)} &= \lim_{n\rightarrow \infty} \textup{Tr}(\sigma_n)^2 \ . 
\end{align*}
The limit is again conveniently done diagrammatically:
\begin{align*}
\scalebox{0.55}{\raisebox{-10ex}{\includegraphics{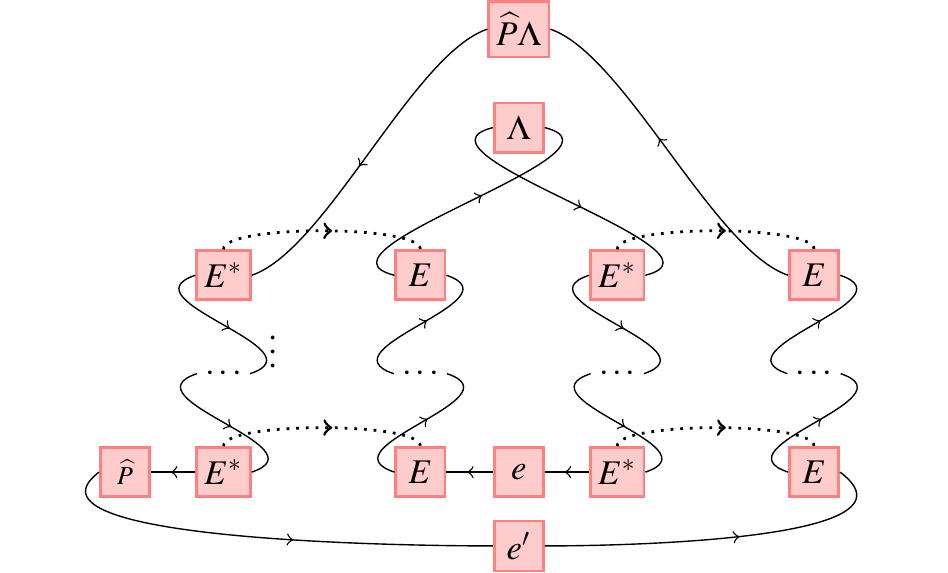} }}  \hspace{-2ex} \longrightarrow  \scalebox{0.55}{\raisebox{-8.5ex}{ \includegraphics{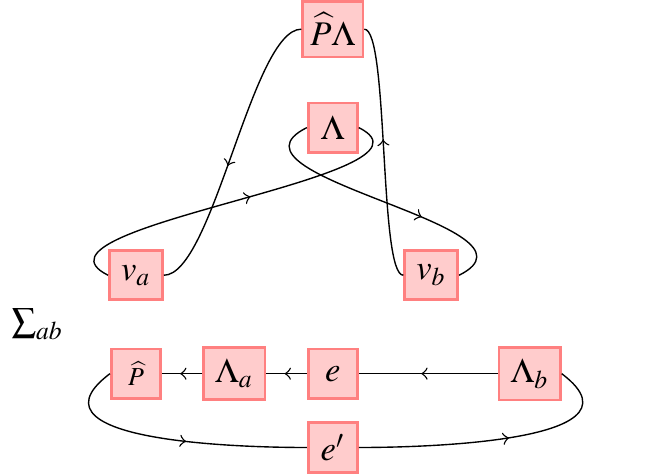}  } } \hspace{-4ex}  \label{eq:Deltalimit}
\end{align*}
Yielding
\begin{align}
e^{-\Delta(\omega)} &= \sum_{ab=0}^\mu \textup{tr}(\Lambda v_b \Lambda v_a) \textup{tr}(e\Lambda_a e \Lambda_b) = \nonumber \\
&= \sum_{ab=0}^\mu \left[ (-1)^{|v_a|} \delta_{ab} \textup{tr}(\Lambda e)^2 \right]^2 = 2^\mu \left[\textup{tr}(\Lambda e)^2\right]^2 \ .
\end{align}
To estimate the speeds of convergence of $Z^C, \ Z^T$ and $\Delta$, it is helpful to consider equations \ref{eq:Zkl},\ref{eq:Zkkdlimit} and \ref{eq:Deltalimit} non-diagrammatically. This will be done here only for the one of them as the other are analogous. Denote by $\mathscr{B}_{24}$ the braiding operator on $\virtHilb \sotimes \virtHilb^* \sotimes \virtHilb \sotimes \virtHilb^*$ that exchanges the second with the fourth factor. Then
\begin{align*}
&\left|e^{-\Delta_n} - e^{-\Delta} \right| = \\
&\quad \quad =  \left| \textup{str}([\rho \sotimes \Lambda] \mathscr{B}_{24} \left( [\grtrans \sotimes \grtrans]^n - [P_1 \sotimes P_1]  \right) (\mathscr{B}_{24})^* [e \sotimes \widehat{P}]) \right| \leq \\
&\quad \quad \leq  \| \rho \sotimes \Lambda\|_1 \|e \sotimes \widehat{P}\| \|(\grtrans^n - P_1) \sotimes \grtrans^n +  P_1 \sotimes(\grtrans^n - P_1)] \| \leq \\
&\quad \quad \leq 2 \| \grtrans^n - P_1\| \ . 
\end{align*}
So all the limits are controlled by $\|\grtrans^n - P_1\|$ and are added up for each of the expressions, which yields exactly the estimates mentioned in the proposition.

\section{Beyond Finite Bond Dimension} \label{chap:Beyond}
Proposition \ref{prop:ZComega} proved that partial transposes allow to construct non-local order parameters for super matrix product states detecting algebraic invariants of the bond algebra:
\begin{align} 
\exp\!\left(\frac{2\pi i}{8} \eta_C(\omega)\right) &:= \lim_{k,\ell \rightarrow \infty} e^{\frac{3}{2}\Delta_{k+\ell}(\omega)} Z^C_{k,\ell}(\omega) \ . \label{eq:etalim} \\
\exp\!\left(\frac{2\pi i}{2} \eta_T(\omega)\right) &:= \lim_{k_1,k_2,d \rightarrow \infty} e^{2\Delta_{k_1+k_2+d}(\omega)} Z^C_{k_1,k_2|d}(\omega) \ . \label{eq:etaTlim}
\end{align}
However, in the class of sMPS states, the corresponding class can be simply read off by investigating the transformation properties of the tensor $E$. It is therefore more interesting to turn the attention to states that are not in such a form. Before taking the thermodynamic limit in equations \ref{eq:etalim} and \ref{eq:etaTlim}, the expressions are manifestly continuous as functions of $\omega$, so it is natural to attempt to take limits. In order to still be able to take the thermodynamic limit, the series considered should be constrained. Here that constraint will be a uniform bound on the correlation length. That is, the invariants will be generalized to the set of well-approximable states of definition \ref{def:wellapprox}.
\begin{myprop} \label{prop:etaTopInv}
	For any $\ell_c < \infty$,
	\begin{itemize}
		\item [(a)] If $C \in G$, $\eta_C \in \mathbb{Z}/8\mathbb{Z}$ is a continuous quantized invariant on $\mathcal{T}^{G}_{\ell_c}$.
		\item [(b)] If $T \in G$, $\eta_T \in \mathbb{Z}/2\mathbb{Z}$ is a continuous quantized invariant on $ \mathcal{T}^{G}_{\ell_c}$.
	\end{itemize}
\end{myprop}
\noindent
Recall that the set of well-approximable states are those that can be approximated by sMPS with a uniform bound on their correlation length. This automatically ensures that the limit state has finite correlation length as well. Removing this assumption would allow for the sequence to approximate e.g. states describing critical systems. In that situation \cite{Calabrese2009}, $\Delta_n(\omega) \sim A \log(n)$ and $Z \sim n^{-3A/2}$, and the argument of $Z$ is ill-defined.
\\
In the following, the analysis will focus on $\eta = \eta_C$, with the understanding that the proof of $\eta_T$ is analogous. 
\begin{mylem} \label{prop:etacont}
	Let
	\begin{align*}
	\exp\!\left(\frac{2\pi i}{8} \eta_n(\omega)\right) &:= e^{\frac{3}{2}\Delta_{k+\ell}(\omega)} Z^C_{k,\ell}(\omega) \ , \  n =\textup{min}(k,\ell) \ .
	\end{align*}
	If $\nu_1,\nu_2$ satisfy an $R_2$-area law, $\Delta_n(\nu_i) \leq \Delta_*$, then
	\begin{align*}
	&\left|e^{\frac{2\pi i}{8} \eta_{n_1}(\nu_1)}  - e^{\frac{2\pi i}{8} \eta_{n_2}(\nu_2)}\right| \leq \\
	&\quad\quad\quad\leq \frac{3 }{2} e^{\frac{5\Delta_*}{2}} \left|e^{-\Delta_{k_1 + \ell_1}(\nu_1)} - e^{-\Delta_{k_2 + \ell_2}(\nu_2)}\right| + \\
	& \quad\quad\quad \quad  + e^{\frac{3\Delta_*}{2}} \left| Z_{k_1\ell_1}(\nu_1) - Z_{k_2\ell_2}(\nu_2)\right| \ . 
	\end{align*}
	In particular,
	\begin{align*}
	&\left|e^{\frac{2\pi i}{8} \eta_{n}(\nu_1)}  - e^{\frac{2\pi i}{8} \eta_{n}(\nu_2)}\right| \leq 5 e^{\frac{5\Delta_*}{2}} \textup{Tr} \left|\sigma_{k + \ell}(\nu_1) - \sigma_{k + \ell}(\nu_2)\right| \ .
	\end{align*}
\end{mylem}
\begin{proof}
	If $z$ is an upper bound of both $x,y$, then the middle-value theorem gives $|x^{-3/2} - y^{-3/2}| \leq \frac{3}{2} |z|^{-5/2} |x-y| $. For the second inequality observe that both $e^{-\Delta_{k+\ell}(\omega)}$ and $Z_{k,\ell}$ are bilinears in the density matrix.
\end{proof}
\noindent
\begin{wrapfigure}{l}{0.25\textwidth}
	\vspace{-4ex}
	\begin{tikzpicture}
	
	
	\node (1) at (0,0) {$\eta_n(\omega_\alpha)$};
	\node (2) at (3,0) {$\eta(\omega_\alpha)$};
	\node (2p) at (3.85,0) {$ \in \mathbb{Z}_8$};
	\node (3) at (0,-2.5) {$\eta_n(\omega)$};
	\node(4) at (3,-2.5) {$\eta$};
	
	
	\draw[->] (1)--(2) node[midway,above] {$\scriptstyle{\textup{Proposition \ref{prop:ZComega}}}$};
	\draw[->] (1) -- (3) node[midway,below,rotate=270] {$\scriptstyle{ \textup{Lemma \ref{prop:etacont} } }$};
	\draw[->] (1) -- (4) node[midway,above,rotate=320] {$\scriptstyle{ \textup{Proposition \ref{prop:etacont} }}$};
	\draw[->] (3) -- (4);
	\draw[->] (2)--(4);
	\end{tikzpicture}
	\vspace{-2.5ex}
	\caption{Limits used in the \\construction of $\eta$.} \label{fig:etalimits}
	\vspace{-4 ex}
\end{wrapfigure}
This lemma and proposition \ref{prop:ZComega} establish the two edges in a square, and the yet to prove proposition \ref{prop:etacont} combines them to their resultant. The procedure is illustrated in figure \ref{fig:etalimits}.
\begin{proof}[Proof of proposition \ref{prop:etacont}.]
	Focus on the $C$-case. Take a state $\omega \in  \mathcal{T}^{G}_{\ell_c}$, and pick an approximating sequence $(\omega_\alpha)_\alpha$ of pure $C$-symmetric super matrix product states. Moreover, by perturbing the sequence slightly, ensure that the transfer operators are diagonalizable. Let $\eta_{n\alpha} = \eta_n(\omega_\alpha)$
	\\
	First prove that $ (\eta_n = \eta_{n,\infty})_n$ is Cauchy. Pick $\epsilon >0$ and write, for $z = \exp(2\pi i/8)$:
	\begin{align*}
	|z^{\eta_{n}} - z^{\eta_{m}}| \leq |z^{\eta_{n}} - z^{\eta_{n\alpha}}| + |z^{\eta_{m}} - z^{\eta_{m\alpha}}| + |z^{\eta_{n\alpha}} - z^{\eta_{m\alpha}}| \ .
	\end{align*}
	By Propositions \ref{prop:Corrwedge} and \ref{prop:ZComega}, the convergence of $(\eta_{n\alpha})_n$ is determined by the correlation length $\ell_c$.
	Hence, the convergence in the last term is uniform in $\alpha$ qua assumption.
	Thus, pick $n$ and $m$ large enough such that the last term is smaller than $\epsilon/3$. Then by lemma \ref{prop:etacont}, each of the first two terms can be made smaller than $\epsilon/3$ by increasing $\alpha$. Hence the following double limit exists:
	\begin{align*}
	\lim_{n \rightarrow \infty} \lim_{\alpha \rightarrow \infty} \eta_{n}(\omega_\alpha) = \eta(\omega) \ .
	\end{align*}
	Next, by proposition \ref{prop:ZComega}, the limit $n \rightarrow \infty$ gives a sequence $(\eta_{\alpha})_\alpha \in \mathbb{Z}_8$. By another $\epsilon/3$-argument one can prove that this sequence converges; to that end note that the $8$th roots of unity have a distance $\sqrt{2 - \sqrt{2}}$ from each other. Hence, pick $\epsilon < \sqrt{2 - \sqrt{2}}/2$. Write:
	\begin{align*}
	\left| z^{\eta_{\alpha}} - z^{\eta_{\beta}}\right| \leq \left| z^{\eta_{n\alpha}} - z^{\eta_{\alpha}}\right|  + \left| z^{\eta_{n\beta}} - z^{\eta_{\beta}}\right| + |z^{\eta_{n\alpha}} - z^{\eta_{n\beta}}| \ .
	\end{align*} 
	Pick $n$ large enough such that either of the first two terms is smaller than $\epsilon/3$. Then, pick $\alpha$ and $\beta$ big enough such that the last term is smaller than $\epsilon/3$. But, since $\eta_\alpha$ is quantized and $\epsilon$ is smaller than half the smallest distance between any of them, the sequence actually has to be constant for large enough $\alpha$, as illustrated in figure \ref{fig:etaconvergence}. Hence
	\begin{align*}
	\lim_{\alpha\rightarrow\infty} \lim_{n \rightarrow \infty} \eta_n(\omega_\alpha) = \eta' \in \mathbb{Z}_8 \ .
	\end{align*}
	It remains to show that $\eta(\omega) = \eta = \eta'$, in particular that it is quantized. Again, write
	\begin{align*}
	| z^{\eta'} - z^{\eta}| \leq |z^{\eta'} - z^{\eta_{\alpha}}| + \left|z^{\eta_\alpha} - z^{\eta_{n\alpha}}\right| + \left| z^{\eta_{n\alpha}} - z^{\eta_{n}}\right| + \left|z^{\eta_n} - z^{\eta}\right| \ ;
	\end{align*}
	and use that each term can be made arbitrarily small. To see that $\eta$ is constant on continuous deformations consider continuous paths $t \mapsto \omega_t$. Pick $t_1,t_2$ and let $\eta_i = \eta(\omega_{t_i})$ and $\eta_{i,n} = \eta_n(\omega_{t_i})$. Then
	\begin{align*}
	| z^{\eta_1} - z^{\eta_2}| = | z^{\eta_1} - z^{\eta_{1,n}}| + | z^{\eta_2} - z^{\eta_{2,n}}| + | z^{\eta_{1,n}} - z^{\eta_{2,n}}| \ .
	\end{align*}
	The first two terms decay exponentially, uniform in $t_1,t_2$, as $n\rightarrow 0$, while the latter, for fixed $n$, goes to zero as $t_1\rightarrow t_2$ by assumption. So $\eta$ is continuous, hence, constant.
\end{proof}
\noindent
The two important ingredients in the proof where the continuity of the invariants at finite volume, and the control on the thermodynamic limit. This allowed to construct two homotopy invariants $\eta_C$ and $\eta_T$. As explained in section \ref{sec:PropertiesofsMPS}, these invariants extract broad information about the bond algebra at finite bond dimension. What about infinite bond dimension? Here, no attempt to construct infinite-dimensional super matrix product representations was undertaken, but these exist \cite{doi:10.1142/S0219025798000351,doi:10.1142/S0129055X13500177}.
\begin{wrapfigure}{r}{0.25\textwidth}
	\includegraphics[scale=0.6]{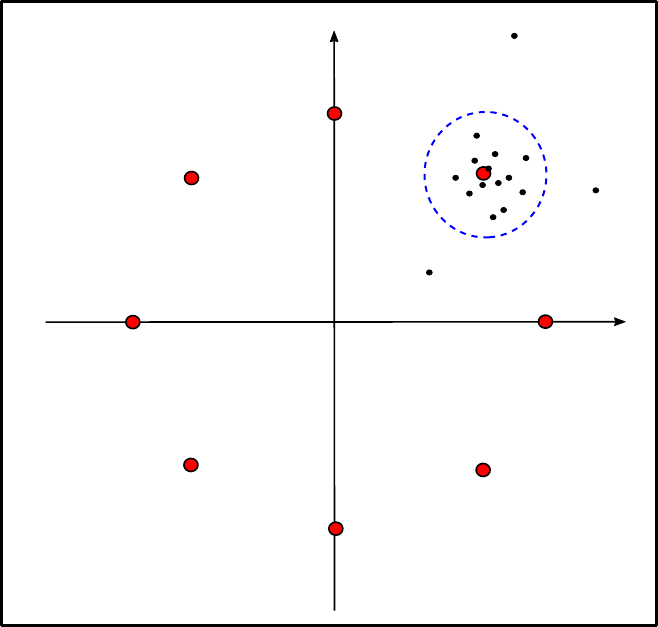}
	\caption{$\eta_n(\omega)$}
	\label{fig:etaconvergence}
\end{wrapfigure}Without explaining the subtleties and modifications that arise in this question, here will be a short outline for an approach to classify $G$-invariant unique gapped ground states. The basic idea is to classify the symmetry representations on the entanglement spectrum, which is extracted from the state by introducing a cut into the system, as in figure \ref{fig:Entanglement}. In the matrix product case, where such spaces are finite-dimensional, this yields an index $I(\omega) \in \mathbb{Z}_2 \times H^1(G,\mathbb{Z}) \times H^2(G,U(1)_\mathfrak{p})$, where the first factor reflects whether a Hilbert space can be associated to a single edge, the second determines the parity of the projective symmetry operations and the last classifies projective $G$-representations with $\mathfrak{p}:G\rightarrow \mathbb{Z}_2$ accounting for anti-unitary operations \cite{Kapustin2018}. Somewhat more surprising, Ogata et al. have since then generalized this idea to all gapped ground states \cite{ogata2019mathbb,ogata2019classification,Bourne_2021,ogata2021classification}. This is an \emph{algebraic} approach that constructs \emph{boundary} invariants -- the boundary can be virtual -- for \emph{all} symmetry groups. In difference, this work produced \emph{analytic} invariants that are defined in the \emph{bulk}, for \emph{certain} features of gapped ground states.
\begin{figure}[H]
	\centering
	\includegraphics[scale=0.6]{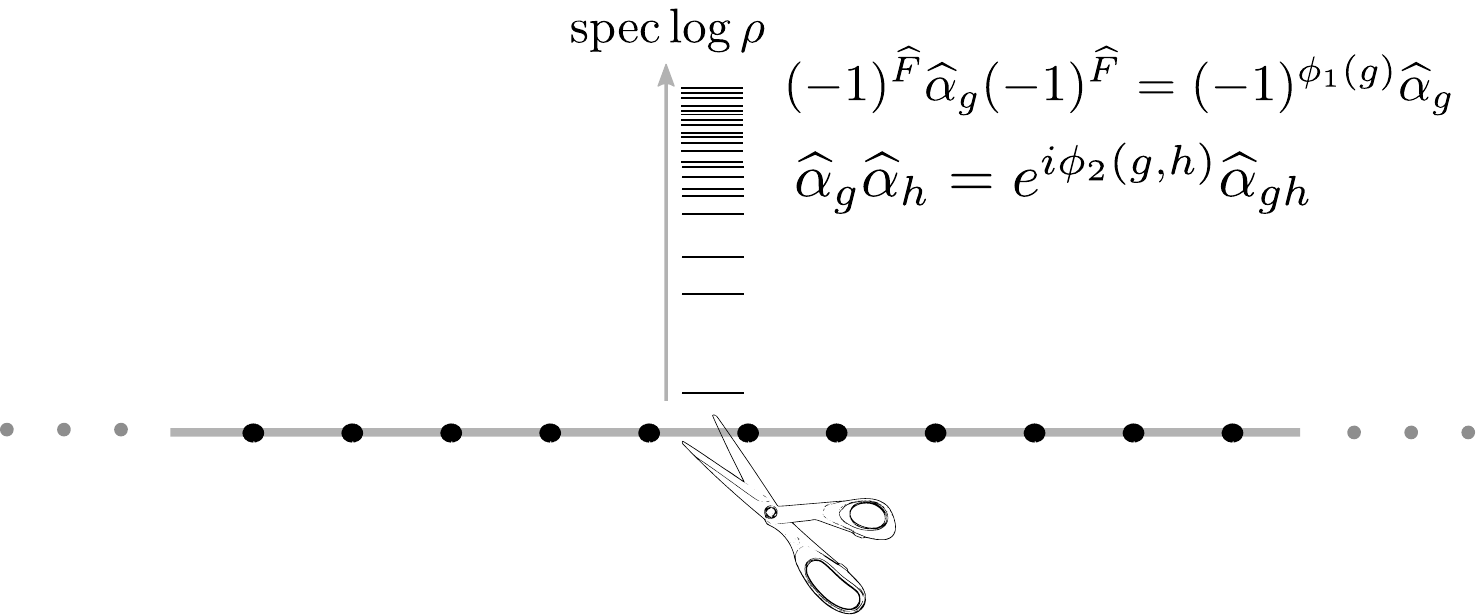}
	\caption{After virtually or really cutting a one-dimensional fermion system, the two halfs transform under projective representations of the symmetry group.	
	} \label{fig:Entanglement}
\end{figure}
\noindent
Algebraic invariants are particular convenient in classificatory work, but if the state is only given indirectly or approximately, they are often difficult to use \cite{pollmann2017symmetry}. For example, the state might have been approximated numerically, or be represented by a (non-topological) quantum field theory. In that context non-local order parameters like the used $\eta_K$ can be useful as they do not depend on the formalism. Furthermore, they can be defined also if the state is not strictly $K$-invariant. While no longer a homotopy invariant, they might still be useful as in recent years people have investigated the disappearance of topological features after the symmetry is broken. The investigations are usually centered on the disappearance of boundary modes \cite{verresen2021quotient}, but $\eta_K$ could give an interesting bulk characterization.

\section{Summary and Outlook}
This paper investigated non-local order parameters in one-dimensional fermionic system defined via anti-unitary operations. The order parameters extract algebraic invariants of the bond algebra, when used on super matrix product states. In particular, when using particle-hole transformations, they extract a real division superalgebra $\mathbb{D}$ such that the bond algebra is a matrix algebra over $\mathbb{D}$. As the definition of the order parameter does not rest on any specific formalism, their domain of definition was enlarged by considering limits of super matrix product states, preserving the symmetry and keeping the correlation length finite in the process. These results, as they were originally inspired by developments in topological field theory, bridge an explicit approach, formulated on the lattice, with the long-range view. Different from unitary symmetries, the considered particle-hole and time-reversal transformations imprint themselves structurally on the bond algebra. It would be interesting to investigate the effect of these structural diversity in more detail; e.g, its implications for correlation functions.
\\
Furthermore, while the exact method used here probably cannot be generalized to higher dimensions, there are certain elements whose exploration in higher dimensions are worthwhile. It should be noted first, that Shiozaki et al. proposed also invariants in higher dimensions \cite{Ryu2018}, which maybe could be verified by different methods. Secondly, the appearances of different \emph{types} of bond algebras could transport to higher dimension. Note that for a tensor network in any dimension, one can associate a vector space, consequentially an algebra, to any edge. Investigation of its algebraic properties could prove useful. 


\begin{acknowledgments}
Partially supported by the Deutsche Forschungsgemeinschaft (DFG, German Research Foundation) – Projektnummer 277101999 – TRR 183 (project C01/XX).
\end{acknowledgments}

\begin{appendix}

	\section{Standard Form of Anti-Unitary Symmetries} \label{appendix:StandardForm}
	A symmetry group generated by an anti-unitary $K$ can be brought in a standard form $U \times A$, where $U$ is an abelian unitary group and either $A \cong \mathbb{Z}_2$ or $A\cong \mathbb{Z}_4$.
	
	\begin{proof}[Proof of lemma \ref{lem:AntiUstandard}]
		Consider the iterates $(K^{2r})_{r \in \mathbb{N}}$. This is an abelian subgroup of the unitaries on $V$. Then this subgroup is either finite (A) or infinite (B).
		\item[(A1)] Consider first the case that there is an odd smallest $n$ such that $K^{2n} = 1$. Then the symmetry is also generated by $\widetilde{K} = K^n$ and $u = K^2$.
		\item[(A2)] Suppose now the smallest $n$ such that $K^{2n} = 1$ is even. Then there is some odd integer $m$ and some integer $r$ such that $2n = 2^rm$, i.e.,  $K^{2^rm} = 1$. Then $K^{2m} = x$ is an involutive unitary. Unitarity is clear, for involutivity consider an eigenvector $v$ with eigenvalue $\lambda$ to $K^{2m}$. Then $\lambda^2 v = x^2(v) = x(\lambda v) = |\lambda|^2 v$ and hence $\lambda = \pm 1$. Then the symmetry group is also generated by $\tilde{K} := K^m$, with $\tilde{K}^4 = 1$ and the unitary $K^2$.
		\item [(B)] If $K^{2r}$ never returns to the identity, it is dense in a $U(1)$ subgroup of unitaries, generated by $X$, say. Hence, $K^2 = \exp(i\alpha X)$ for some real $\alpha$. Then $K$ and $X$ anticommute since  $K \exp(i\alpha X) = K^3 = \exp(i\alpha X) K$. Thus, define a new anti-unitary $\widetilde{K} := \exp(-i\alpha X/2) K$, squaring to $+1$.
	\end{proof}
	\noindent
	
	\section{Transfer Matrix Spectra and Correlation lengths} \label{chap:sMPSSpectra}
	Here is a proof of proposition \ref{prop:Corrwedge}, which was the statement that for any sMPS $\omega$ there is a constant $1 \geq C>0$ s.t.:
	\begin{align*}
	C \| \grtrans^k - P_e\|_{A} \leq \textup{Corr}_{\omega}(k) \leq \| \grtrans^k - P_e\|_{A} \ .
	\end{align*}
	\begin{proof}
		Start with the second inequality. Since $\grtrans_\mathcal{O}$ can be written in terms of the isometry \ref{eq:Undef}, it is positive. As the norm of a positive map is given by its value at the identity,
		\begin{align}
		\|\grtrans_\mathcal{O}(e)\| \leq \|\mathcal{O}\|\|\grtrans_{1}(e)\| = 1 \ ; \label{eq:grObound}
		\end{align}
		and similar for $\|(\grtrans_{\mathcal{O}})^\prime(\lambda)\|$ where $(\grtrans_{\mathcal{O}})^\prime(\lambda) = (-1)^{|\lambda||\mathcal{O}|} \lambda \circ \grtrans_{\mathcal{O}}$.
		Denote by $S^k$ the translation by $k$ sites. Fix local operators $\mathcal{O}_1,\mathcal{O}_2$ with disjoint support. Then:
		\begin{align*}
		|\omega(\mathcal{O}_1 S^k( \mathcal{O}_2))& - \omega(\mathcal{O}_1) \omega(\mathcal{O}_2)| = \\
		&= |\textup{tr}(\rho_E \grtrans_{\mathcal{O}_1} \circ (\grtrans^k - P_e) \circ \grtrans_{\mathcal{O}_2}(e))| \leq \\
		&\leq \|\grtrans^k - P_e\|_{A} \|\grtrans_{\mathcal{O}_2}(e)\| \|(\grtrans_{\mathcal{O}_1} )^\prime(\lambda)\| \leq \\
		&\leq\|\grtrans^k - P_e\|_{A} \|\mathcal{O}_1\| \|\mathcal{O}_2\| \ .
		\end{align*}
		The subscript $A$ indicates that the supremum is taken over the bond algebra only. This proves the inequality on the right hand site.
		\\
		Thus turn to the other inequality. The duality of Schmidt norms gives:
		\begin{align}
		\| \grtrans^k - P_e \|_{A} &= \sup_{a,b\in A} \{ |\textup{tr}(b (\grtrans^k - P_e)(a))| \ | \ \|a\|=1 \ ,  \|b\|_1 = 1    \} \ , \label{eq:grtranskab}
		\end{align}
		where $\|x\|_1 := \textup{tr}|x|$ is the $1$-norm on $A$. Use that for a pure state the Kraus algebra is closed under conjugation so that both $\mathcal{O} \mapsto \grtrans_{\mathcal{O}}(e)$ and $\mathcal{O} \mapsto (\grtrans_\mathcal{O})^\prime(\lambda)$ are surjective. Find $\mathcal{O}_a,\mathcal{O}_b$ s.t. $a = \grtrans_{\mathcal{O}_a}(e)$ and $\textup{tr}(b \, \cdot) = (\grtrans_{\mathcal{O}_b})^\prime(\lambda)$. Then:
		\begin{align*}
		\frac{ |\textup{tr}(b (\grtrans^k - P_e )(a))|}{\|a\| \|b\|_1} &= \left(\frac{\|\mathcal{O}_a\| \|\mathcal{O}_b\|}{\|a\| \|b\|_1}\right) \times \\
		&\times \left[\frac{|\omega(\mathcal{O}_b S^k( \mathcal{O}_a)) - \omega(\mathcal{O}_b) \omega(\mathcal{O}_a)|}{\|\mathcal{O}_a\| \| \mathcal{O}_b\|} \right] \ .
		\end{align*}
		The expression in the square brackets is bounded by $\textup{Corr}_\omega(k)$. To prove the proposition, study the ways of choosing $\mathcal{O}_a$ for given $a \in A$. It is sufficient here to restrict to the unit sphere in $A$. First of all, consider the sequence of maps $f_n: \mathscr{L}(\Lambda(V)^{\sotimes n}) \rightarrow A$ given by $f_n(\mathcal{O}) = \grtrans_\mathcal{O}(e)$. Since $\textup{Im}(f_n) \subset \textup{Im}(f_{n+1})$ and $A$ is finite-dimensional, there is a $n_*$ s.t. $\textup{Im}(f_{n_*}) = A$. Denote $f_{n_*} \equiv f$. Then introduce a function $a\mapsto r_e(a)$, defined for non-zero $a\in A$:
		\begin{align*}
		r_e(a) := 
		\textup{inf}\left\{\|\mathcal{O}\| \ : \ f(\mathcal{O}) = a \right\} \ .
		\end{align*}
		Define a similar function $r_\lambda$ for $\mathcal{O} \mapsto (\grtrans_\mathcal{O})^\prime(\lambda)$. Both $r_e$ and  $r_\lambda$ are, when restricted to the unit sphere in $A$, functions from a compact set into $\mathbb{R}$, hence bounded. Denote these bounds by $C_e$ and $C_\lambda$. Then
		\begin{align*}
		\| \grtrans^k - P_e \|_{A} \leq C_e C_\lambda \textup{Corr}_\omega(k) \ .
		\end{align*}
		Note that $C = (C_eC_\lambda)^{-1}$ satisfies $0 < C \leq 1$ as equation \ref{eq:grObound} gives a lower bound for $r_e$ as $r_e(a) \geq \|a\|$.	
	\end{proof}
	\noindent

\end{appendix}


\end{document}